\documentclass[11pt]{article}

\usepackage{amsmath, amsfonts, amsthm, amssymb}
\usepackage{ifthen}
\usepackage{multirow,url}
\usepackage{graphicx}
\usepackage [table]{ xcolor }
\usepackage{subfig}
\usepackage{multirow}
\usepackage[breaklinks = true]{hyperref}
\usepackage{tikz}
\usetikzlibrary{arrows,decorations.pathmorphing,decorations.footprints,fadings,calc,trees,mindmap,shadows,decorations.text,patterns,positioning,shapes,matrix,fit}

\newcommand{\ba}{\begin{eqnarray}}
\newcommand{\ea}{\end{eqnarray}}
\newcommand{\Be}{\mathsf{Be}}
\newcommand{\Bin}{\mathsf{Bin}}

\def\EE{\mathbb{E}}
\def\PP{\mathbb{P}}

\def\NN{\mathbb{N}}
\def\RR{\mathbb{R}}

\def\bfdplusn{\mathbf{d}^+_n}
\def\bfdminusn{\mathbf{d}^-_n}
\def\dplusn{d^+_n}
\def\dminusn{d^-_n}
\def\dplus{d^+}
\def\dminus{d^-}
\def\en{e_n}

\def\En{E_n}

\def\gamman{\gamma_n}
\def\config{G^*_n(\bfdminusn, \bfdplusn)}

\newtheorem{theorem}{Theorem}[section]
\newtheorem{remark}[theorem]{Remark}

\newtheorem{example}[theorem]{Example}
\newtheorem{assumption}[theorem]{Assumption}
\newtheorem{definition}[theorem]{Definition}
\newtheorem{corollary}[theorem]{Corollary}
\newtheorem{proposition}[theorem]{Proposition}

\newtheorem{lemma}[theorem]{Lemma}

\def\ind{{\rm 1\hspace{-0.90ex}1}}
\def\NN{\mathbb{N}}

\def\ind{{\rm 1\hspace{-0.90ex}1}}
\def\EE{\mathbb{E}}
\def\PP{\mathbb{P}}
\def\PP{\mathbb{P}}

\def\en{e_n}
\def\En{E_n}
\def\gamman{\gamma_n}
\def\CM{G^*_n(\mathbf{e}_n)}
\def \CMM{\tilde{G}_n(\mathbf{e}_n, \mathbf{\gamma}_n)}
\renewcommand{\epsilon}{\varepsilon}

\tikzstyle{vertex}=[circle,fill=blue!25,minimum size=27pt,inner sep=0pt]
\tikzstyle{selected vertex} = [circle, vertex, fill=red!24,minimum size=27pt,inner sep=0pt]
\tikzstyle{edge} = [draw,thick,double, ->]
\tikzstyle{weight} = [font=\small]
\tikzstyle{selected edge} = [draw,line width=5pt,-,red!50]
\tikzstyle{ignored edge} = [draw,line width=5pt,-,black!20]
\tikzstyle{random edge} = [draw,line width=2pt,-,green!40]
\tikzstyle{inverse edge} = [draw,thick,<-]
\tikzstyle{undirected edge} = [draw,thick,-]
\tikzstyle{large}=[ellipse,fill=blue!50,minimum size=15pt,inner sep=0pt]
\tikzstyle{medium}=[ellipse,fill=blue!30,minimum size=15pt,inner sep=0pt]
\tikzstyle{small}=[ellipse,fill=blue!10,minimum size=15pt,inner sep=0pt]
\tikzstyle{zero}=[ellipse,fill=red!24,minimum size=15pt,inner sep=0pt]
\tikzstyle{large loss} = [draw,line width=6pt,-,red!50]
\tikzstyle{medium loss} = [draw,line width=4pt,-,red!50]
\tikzstyle{small loss} = [draw,line width=2pt,-,red!50]

\tikzset{mylabel/.style={text width=4cm, text centered}
}

\tikzset
{
   	SensorNodeStyle/.style =
	{
		circle,									
		minimum size	= 4.5mm,				%
		rotate			= 0,					
		scale			= 1.0,					
		thick,									
		fill			= green!10,				
		text			= black,				
		draw			= black,				
		font			= \scriptsize,				
		inner xsep		= 0mm,					
		inner ysep		= 0mm,					
		text height		= 0.2cm,
		text depth		= 0.12cm,
	}
}

\tikzset
{
	BayStationStyle/.style =
	{
		rectangle,									
		rounded corners	= 1mm,					
		minimum size	= 4.5mm,				%
		rotate			= 0,					
		scale			= 1.0,					
		thick,									
		fill			= black!30,				
		text			= black,				
		draw			= black,				
		font			= \scriptsize,				
		text centered,							
		inner xsep		= 0mm,					
		inner ysep		= 0mm					
	}
}

\tikzset
{
	FittingStyle/.style =
	{
		shape = rectangle,							
		rounded corners	= 3pt,					
		minimum height	= 0.15\textwidth,		
		scale			= 1,					
		thick,									
		draw			= red,				
		inner xsep		= 0mm,					
		inner ysep		= 0mm					
	}
}

\tikzset
{
	GenericNodeStyle/.style =
	{
		shape = rectangle,							
		rounded corners	= 3mm,					
		minimum height	= 1cm,					
		minimum width	= 2cm,					
		scale			= 1.0,					
		thick,									
		fill			= green!10!white,				
		draw			= green,				
		inner xsep		= 3mm,					
		inner ysep		= 3mm					
	}
}

\tikzset
{
	NormalNodeStyle/.style =
	{
		shape = circle,							
		minimum size	= 20,					%
		rotate			= 0,					
		scale			= 1.0,					
		thick,									
		text			= black,				
		draw			= black,				
		font			= \small,				
		text centered,							
		inner xsep		= 0,					
		inner ysep		= 0						
	}
}

\tikzset
{
	BridgeNodeStyle/.style =
	{
		circle,									
		minimum size	= 20,					%
		rotate			= 0,					
		scale			= 1.0,					
		thick,									
		fill			= black!30,				
		text			= black,				
		draw			= black,				
		font			= \small,				
		text centered,							
		inner xsep		= 0,					
		inner ysep		= 0						
	}
}

\tikzset
{
	WarningTextStyle/.style =
	{
		rectangle,						
		rounded corners	= 0.6cm,		%
		minimum size	= 2cm,			%
		rotate			= 0,			
		scale			= 1.0,			
		thick,							
		fill			= red!10,		
		text			= red!10!black,	
		draw			= red,			
		font			= \large,		
		text centered,					
		text width		= 10cm,			
		inner xsep		= 0.5cm,		
		inner ysep		= 0.5cm			
	}
}

\tikzstyle{sNormalBlockStyle} =
[
	draw,
	rectangle,
	rounded corners	= 0.1cm,
	fill			= blue!20,
	minimum height	= 3em,
	minimum width	= 6em,
]

\tikzstyle{sSumBlockStyle} =
[
	shape			= circle,
	draw,
	fill			= blue!20,
]

\tikzstyle{sArrowsStyle} =
[
	thick,
	color	= black,
	-latex
]

\tikzstyle{sLinesStyle} =
[
	thick,
	color	= black,
	-
]

\tikzstyle{sTextBlockStyle} =
[
	draw,
	rectangle,
	drop shadow,
	rounded corners	= 0.1cm,
	fill			= blue!10,
	thick,
	inner xsep		= 0.2cm,		
	inner ysep		= 0.2cm			
]

\tikzfading	
[
	name			= middle,
	top color		= transparent!100,
	bottom color	= transparent!100,
	middle color	= transparent!00,
]

\tikzstyle{sCoalFiredPlant} =
[
	shape			= rectangle,
	minimum height	= 0.5cm,
	minimum width	= 0.5cm,
	rounded corners	= 0.1cm,
	fill			= black!20,
	draw			= black,
	line width		= 0.1cm,
	inner xsep		= 0.2cm,
	inner ysep		= 0.2cm
]

\graphicspath{{Images/}}

\begin{document}
\title{Resilience to contagion in financial networks}
\date{}
\author{Hamed Amini\footnote{\'Ecole Normale Superi\'eure, Paris - INRIA Rocquencourt, Hamed.Amini@ens.fr} \and  Rama Cont\footnote{Columbia University - Universit\'e Paris VI (CNRS), Rama.Cont@columbia.edu}  \and Andreea Minca\footnote{Universit\'e Paris VI - INRIA Rocquencourt, Andreea.Minca@inria.fr}}
\maketitle
\begin{abstract}
Propagation of balance-sheet or cash-flow insolvency across financial institutions may be
modeled as a cascade process on a network representing their
mutual exposures. We derive rigorous asymptotic results for the
magnitude of contagion in a large financial network and give an
analytical expression for the asymptotic fraction of defaults, in
terms of network characteristics. Our results extend previous
studies on contagion in random graphs to
inhomogeneous directed graphs with a given degree sequence and
arbitrary distribution of weights. We introduce a
criterion for the resilience of a large financial network to the
insolvency of a small group of financial institutions and  quantify how contagion  amplifies  small shocks to the network. Our results
emphasize the role played by ``contagious links'' and show
that institutions which contribute most to network instability in case of
default have both large connectivity and a large fraction of contagious links. The asymptotic results show good agreement with  simulations for networks with  realistic sizes.

\vspace{0.4cm}
\noindent \textbf{Keywords:} systemic risk, default contagion, random graphs,
macro-prudential regulation.
\end{abstract}
\newpage

\section{Introduction}

The recent financial crisis has highlighted the complex nature of linkages between financial institutions. From balance-sheet exposures, to more opaque obligations related to over-the-counter derivatives like credit default swaps, such linkages propagated and amplified financial distress. Initial losses in one asset class --mortgage backed securities-- turned into losses that threatened the stability of the whole financial system. More than $370$ of the almost $8000$ US bank companies have failed since $2007$. This was clearly an episode of large default contagion, if we compare to a number of $30$ defaults in the period $2000 - 2004$, and no defaults occurred in the period $2005-2006$.

The acknowledgement of different types of interbank connections and the associated contagion mechanisms led to an increased advocacy to account for network effects when discussing regulatory requirements \cite{haldanemay11, imf09}.
In the growing body of work dedicated to systemic risk, several distress propagation mechanisms have been pinpointed, including primarily two types of insolvency: balance-sheet insolvency and cash-flow insolvency. A bank is said to be balace-sheet insolvent  if the value of its liabilities exceeds the value of its assets and, it is said to be cash-flow insolvent if it cannot meet its contractual payment obligations arrived at maturity.

Cascades of insolvencies can be understood as domino effects and are a type of potent financial contagion that is quantifiable.

Most investigated, balance-sheet insolvency contagion can be described as follows: party $A$ has a balance sheet exposure to party $B$, where we understand by exposure the maximum loss incurred by $A$ on its balance sheet claims upon the default of $B$.  If $B$ defaults, the capital of party $A$ must absorb the corresponding loss. If the capital cannot withstand the loss, party $A$ becomes \emph{balance sheet insolvent }.

Another contagion mechanism that can be described in similar terms is represented by \emph{cascades of cash-flow insolvencies}. Over-the-counter derivatives markets are prone to such type of cascades. Indeed, in these markets parties deal directly with one another rather than passing through an exchange. As such, they are subject to the risk that the other party does not fulfill its payment obligations. Consider two parties $A$ and $B$, such that $A$ has a receivable from party $B$ upon the realization of some event. If $B$ does not dispose of enough liquid reserves, it will default on the payment. Now consider that $B$ has entered an off-setting contract with another party $C$,  hedging its exposure to the random event. If $C$ is cash-flow solvent, then the payment will flow through the intermediary $B$ and reach $A$. However, if $C$ is cash-flow insolvent and defaults, then the intermediary $B$ might become cash-flow insolvent if it depends on receivables from $C$ to meet its payment obligations to $A$.
As it turns out, the length of such chains of intermediaries in certain over-the-counter markets, like the credit default swap market, is significant \cite{ramacds10, mincacont10}, thereby increasing the probability of a cascade of cash-flow insolvencies.\footnote{ The trigger of a cash-flow insolvency cascade in a chain of intermediaries may be  a counterparty that is not contractually required to post collateral when all other entities have this obligation.}


Some of previous work, mostly in the economics and sociology literature, investigates cascades on networks in a generic context;  relevant references include Morris \cite{morris00},  Kleinberg \cite{Kleinberg07}, Jackson and Yariv \cite{JacksonYariv07} and Watts \cite{watts02}. These models consider, in one form or another, a mechanism by which an agent decides to adopt one of two states depending on the state of its neighbors and a threshold which measures its susceptibility to this direct influence. Propagation of insolvency in banking networks fall under the irreversible version of this model \cite{watts02, Kleinberg07} ; default is not reversible, unlike the case of agents playing a network game who can revise their decisions \cite{morris00}.  The default threshold of a given bank depends on its level of capital (resp. liquidity reserve), the state of balance sheet (resp.  cash-flow) solvency of its direct counterparties and the linkages to them. More recent work studies network formation as a result of the interplay between benefits from creating links and negative network externalities \cite{Blume11}.

In the finance literature, in the context of banking systems, contagion effects and network externalities have been  investigated in both in theoretical \cite{allen00, Battiston09, eisenberg01, gai10, nier07} and empirical studies \cite{cont10c, elsinger06a, upper02, muller06}.
Network externalities are --implicitly or explicitly-- present in various early discussions of  systemic risk (see e.g., Hellwig \cite{hellwig1995}, Kiyotaki and Moore \cite{kiyotaki2002},
Rochet and Tirole \cite{rochettirole}) through the interlinkages between balance sheets.
Allen and Gale \cite{allen00}  pioneered the use of network models in the  study of the  stability of a system of interconnected financial institutions. Their results were extended in various directions by Lagunoff and Schreft \cite{lagunoff01} and Leitner \cite{leitner05}.

One of the central problems tackled in this literature is understanding the relation between the cascading behavior of the network and the underlying topology: Is the network such that the state of a small number of nodes will propagate to a large number of nodes, or will contagion die out quickly?

The discussion  was either dominated by highly stylized networks, whose structure turns out to be quite different from the heterogenous structure of real networks --financial networks are particularly heterogenous, as many of the empirical studies \cite{boss04, Soramaki07, cont10c}  make a clear case-- or the results were heuristic in nature and based on mean field approximations \cite{watts02, gai10}. Notable exceptions come from the random graphs literature, where cascade models are investigated on graphs with given degree sequences \cite{amini, Lelarge11}.
One crucial aspect that does not appear in the previous literature is the heterogeneity of weights:  whether interpreted as exposures or receivables, these linkages carry weights with a heavy tailed distribution. This point --corroborated by simulations \cite{cont10c}-- prevents from reducing the analysis of contagion in banking networks to the case where a node's aggregate exposure is distributed equally across counterparties as in \cite{gai10, May10}.

In light of insights coming from empirical studies and simulations, we redefine the problem as relating the cascading behavior of financial networks both to the local properties of the nodes and to the underlying topology of the network. Since balance-sheet and cash-flow insolvency cascades are similar from a mathematical modeling point of view, the challenge lies not so much in analyzing a model that is flexible enough  to represent both these types of insolvency cascades, but in proposing a model that can mimic the empirical properties of these different types of networks and that is tractable enough to be able to prove theorems about the cascading behavior.

Our problem is set form the point of view of a regulator who observes the network.
Our primary goal is not to identify nodes posing the highest systemic risk. Clearly, when one knows the entire network and assuming the network is fairly small, those could be identified by extensive simulations. Alternatively, for threshold models of contagion  in large networks one could use approximation algorithms to find sets of most influential nodes \cite{kempe03, mossel07}. Indeed, problems like the Influence Maximization Problem have been shown NP-hard to approximate within a factor $1 - 1/e + \epsilon$ for all $\epsilon > 0$ \cite{kempe03}, but under certain sub-modularity conditions there exists a greedy  $1 - 1/e - \epsilon$ - approximation algorithm for this problem \cite{mossel07}.

The fundamental question that we tackle, unanswered so for realistic networks, is how to identify the features that make nodes systemically important.
The obvious purpose is the need to set rules that would mitigate such features and consequently systemic risk.

Our approach is to consider an ensemble of networks in which one can prescribe each node's connectivity and characteristics that are relevant to the respective cascade mechanisms: balance sheet insolvency cascades depend on capital ratios and the asset side of the balance sheets; cash-flow cascades depend on liquidity reserves and cash flows related to positions in the trading book.
When the number of banks is large, cascades on networks belonging to this ensemble behave in a way dictated by the prescribed characteristics.

\subsection{Summary}
In this paper we develop techniques for analyzing default cascades in random weighted directed networks with arbitrary degree sequences, in which a set of local features can be prescribed for each node in the network.

Our contribution is to  derive rigorous asymptotic results for the
magnitude of contagion in  such  networks and give an
analytical expression for the asymptotic fraction of defaults.
Our results apply to a wide variety of topologies and provide analytical insights into the nature of the relation between network structure, local characteristics of nodes and contagion in large-scale networks.

For simplicity, we formulate our results in terms of {\it balace-sheet insolvency cascades} in a network of financial institutions with interlinked balance sheets, where losses flow into the asset side of the  balance sheets. Similar techniques may be used for analyzing cascades of cash-flow insolvency in over-the-counter markets, as briefly discussed above and detailed in \cite{mincathesis11}.  From now on, we refer to balance-sheet insolvency simply as insolvency.

Our proof is  based on a coupling argument: We construct a related multigraph --a weighted configuration model-- which leads to the same number of defaults as in the original contagion process but is easier to study because of its independence properties.  The contagion process in this model may then be described by a  Markov chain. Generalizing the differential equation method of Wormald \cite{Worm95} to the case where the dimension of the Markov chain depends on the size of the network we show that, as the network size increases, the rescaled Markov chain converges in probability to a limit described by a system of ordinary differential equations, which can be solved in closed form.
This enables us to obtain analytical results on the final fraction of defaults in the network.

These results generalize previous ones on diffusions in random graphs with prescribed degree sequence to the case of inhomogeneous and weighted random directed graphs with arbitrary degree sequences. Related problems are the problem of existence of a giant component in random graphs \cite{coopfri04,Molloy98thesize} and bootstrap percolation problem. Bootstrap percolation process  \footnote{A \emph{bootstrap percolation process} on a graph $G$ is an ``infection" process which evolves in rounds. Initially, there is a subset of infected nodes and in each subsequent round each uninfected node which has at least $r$ infected neighbors becomes infected and remains so forever (The parameter $r\geq 2$ is fixed.).} is a very simple models of difusions which have been studied on a variety of graphs, such as trees~\cite{BPP06}, grids~\cite{holroyd03, BBDM2010}, hypercubes~\cite{BB06}, as well as on several distributions of random graphs~\cite{amini, amini-nn, amfou, ar:JLTV10, balpit07}.


Another important result of our work --and probably the most important from the regulatory point of view-- is the introduction of a measure of resilience of a financial network
to small initial shocks. The contribution of each node to systemic risk is quantifiable in terms of its connectivity and local characteristics.
Our measure may be used as a tool for stress testing the resilience of interbank networks in a decentralized way \cite{amini10b} and as an assessment tool of the capital adequacy of each bank with respect to its exposures.



\subsection{Outline}
The paper is structured as follows. Section \ref{sec:preliminaries} introduces a model for a network of  financial institutions
 and describes a mechanism for default contagion in such a network.
 Section \ref{sec:main} gives our main result on the asymptotic magnitude of contagion.
 Section \ref{sec:resilience} uses this result to define a measure of resilience for a financial network: We show that when this indicator of resilience crosses a threshold, small initial shocks to the network --in the form of the  exogenous default of a small set of  nodes-- may generate a large-scale cascade of failures, a signature  of {\it systemic risk}.
Section \ref{sec:results} illustrates,  through concrete examples, how the resilience measure allows us to quantify and predict the outcome of contagion on one sample network generated from a random network model that mimics the properties of a real interbank exposure network analyzed in \cite{cont10c}.
We observe that networks with the same
average connectivity may  amplify initial
shocks in very different manners and their resilience to contagion can vastly differ. In particular, the relation between `connectivity' and 'contagion' is not monotonous.
Technical proofs are given in Appendix \ref{sec:proofs}.

\section{A network model of default contagion}
In this section, we first introduce a model of a financial network, then describe the default cascade on this network, and finally the probabilistic setting we use throughout the paper.
\label{sec:preliminaries}
\subsection{Counterparty networks}
\label{sec:model}
Interlinkages across balance sheets of financial institutions may be modeled by a weighted directed graph $G =(V, \mathbf{e})$ on the vertex set $V = \{1, \dots, n\} = [n]$, whose elements
represent financial institutions. The \emph{exposure matrix} is given by $\mathbf{e} \in \RR^{n \times n}$, where the $ij$-{th} entry $e(i,j)$ represents the exposure (in monetary units) of institution $i$ to institution $j$.
Table \ref{balancesheet.tab} displays a stylized balance sheet of
a financial institution.
The interbank assets of an institution $i$ are given by $$A(i) := \sum_j e(i,j).$$ Note that  $\sum_j e(j,i)$ represents the interbank liabilities of $i$.
  In addition to these interbank assets and liabilities, a bank may hold other assets and liabilities (such as deposits).

\begin{table}
  \centering
  \begin{tabular}{|c|c|}
    \hline
    Assets & Liabilities \\
    \hline
   Interbank assets  & Interbank liabilities \\
    $\sum_j e(i,j)$ & $\sum_j e(j,i)$ \\
     &  Deposits\\
      &  $D(i)$ \\
\hline
    Other  & Net worth \\
     assets &  \\
    $x(i)$ &  $c(i) = \gamma(i)A(i)$\\
     & \\
    \hline
  \end{tabular}
  \caption{Stylized balance sheet of a bank.}\label{balancesheet.tab}
\end{table}
The net worth of the bank, given by its \emph{capital} $c(i)$,
represents its capacity for absorbing losses while remaining solvent. We will refer to the ratio
$$\gamma(i) := \frac{c(i)}{A(i)}$$  as the ``capital ratio" of institution $i$, although technically it is the ratio of capital to interbank assets and not total assets.
\begin{center}\it
An institution is {\it insolvent} if its net worth is negative or
zero, in which case we set $\gamma(i)=0$.
\end{center}
\begin{definition}[Financial network]\rm
A financial network $\mathbf{(e, \gamma)}$ on the vertex set $V = [n]$ is defined by
\begin{itemize}
\item a matrix of exposures $\{e(i,j)\}_{1 \leq i,j \leq n}$,
\item a set of capital ratios $\{\gamma(i)\}_{1 \leq i \leq n}.$
\end{itemize}
\end{definition}

\noindent In this network, the {\it
in-degree} of a node  $i$ is given by
$$ d^-(i) := \#\{j\in V \mid \ e(j,i) > 0\} ,$$
which represents the number of  nodes exposed to  $i$,
while its  {\it out-degree}
$$ d^+(i) := \#\{j\in V \mid \  e(i,j) > 0\} $$
represents the number of institutions $i$ is exposed to.

The set of initially
insolvent institutions is represented by
$$\mathbb{D}_0(\mathbf{e},\gamma) = \{ i \in V \mid \ \gamma(i) = 0 \} .$$

The next section defines the default cascade triggered by nodes in $\mathbb{D}_0(\mathbf{e},\gamma) $.
\subsection{Default contagion}\label{sec-def-dyn}
In a network $\mathbf{(e, \gamma)}$ of counterparties,
the default of one or several nodes may lead to the insolvency of
other nodes, generating a {\it cascade} of defaults.

Starting from the set of initially insolvent institutions $\mathbb{D}_0(\mathbf{e},\gamma)$
which represent  {\it fundamental defaults}, we define a contagion process as follows.

Denoting by $R(j)$ the recovery rate on the assets of $j$ at default,  the
default of $j$ induces a loss equal to $(1-R(j))e(i,j)$ for its
counterparty $i$. If this loss exceeds the capital of $i$, then
$i$ becomes in turn insolvent. Recall that $c(i) = \gamma(i) A(i)$. The set of nodes which become insolvent  due to
their exposures to initial defaults is
$$\mathbb{D}_1(\mathbf{e},\gamma) = \{i \in V  \mid \ \gamma(i) A(i) < \mathop{\sum}_{j \in \mathbb{D}_{0}}(1 - R(j))e(i,j)\}.$$
This procedure may be iterated to define the {\it default cascade} initiated by a set of initial defaults.

\begin{definition}[Default cascade]\rm
\label{domino}
Consider a financial network $\mathbf{(e, \gamma)}$ on the vertex set $V = [n]$. Set $\mathbb{D}_0(\mathbf{e, \gamma}) = \{i \in V \mid \ \gamma(i) = 0\}$ of initially insolvent institutions.
The increasing sequence $(\mathbb{D}_k(\mathbf{e},\gamma),k\geq 1)$ of subsets of $V$ defined by
$$\mathbb{D}_k(\mathbf{e, \gamma}) = \{i \in V \mid \ \gamma(i)A(i) < \mathop{\sum}_{j \in \mathbb{D}_{k-1}(\mathbf{e, \gamma})}(1 - R(j))e(i,j)\}$$
is  called the \emph{default cascade} initiated by $\mathbb{D}_0(\mathbf{e},\gamma)$.
\end{definition}

Thus $\mathbb{D}_k(\mathbf{e, \gamma})$ represents the set of institutions whose capital is insufficient
to absorb losses due to defaults of institutions in
$\mathbb{D}_{k-1}(\mathbf{e, \gamma})$.

It is easy to see that, in a network  of size $n$, the cascade ends after at most $n-1$ iterations. Hence, $\mathbb{D}_{n-1}(e,\gamma)$ represents the set of all nodes which become insolvent starting from the initial set of defaults $\mathbb{D}_0(e,\gamma)$.

\begin{definition}\rm
\label{def:domino} Consider a financial network $\mathbf{(e, \gamma)}$ on the vertex set $V = [n]$. The \emph{fraction of defaults} in the network
$\mathbf{(e, \gamma)}$ (initiated by $\mathbb{D}_0(\mathbf{e},\gamma)$) is given by
$$\alpha_n(\mathbf{e, \gamma}) := \frac{|\mathbb{D}_{n-1}(\mathbf{e, \gamma})|}{n}.$$
\end{definition}

The recovery rates $R(i)$ may be exogenous or, as in Eisenberg and
Noe \cite{eisenberg01}, determined endogenously by redistributing
assets of a defaulted entity among debtors, proportionally to
their outstanding debt. As noted in \cite{upper10,cont10c}, the
latter scenario is too optimistic since in practice liquidation
takes time and assets may depreciate in value due to fire sales
during liquidation. As argued in \cite{cont10c,elsinger06a}, when examining  the short term consequences of default, the
most realistic assumption on recovery rates is zero: Assets held with a defaulted counterparty are frozen until liquidation takes place, a process which can in practice take months to terminate.

\begin{center}\it
For simplicity, we assume from now on that recovery rates are constant for all institutions:
$R(i) = R, \ \forall i \in V.$
\end{center}

\subsection{A random network model}
\label{sec-model-graph}

Empirical studies on interbank exposures \cite{boss04,cont10c} show such networks to have a complex and heterogeneous structure
characterized by heavy-tailed (cross-sectional) distributions of degrees and exposures.

Given a description of the large-scale structure of the network in statistical terms, it is natural to model the network as a {\it random graph}
whose statistical properties correspond to these observations.

Consider a sequence $(\mathbf{e}_n,\gamman)_{n\geq 1}$ of  financial
networks, indexed by the number of nodes $n$, where $\bfdplusn=\{\dplusn(i)\}^n_{i=1}$ (resp.
$\bfdminusn=\{\dminusn(i)\}^n_{i=1}$) represents the sequence of   in-degrees (resp. out-degrees) of nodes in $\mathbf{e}_n$.
We now construct a random network  $\mathbf{E}_n$ such
that   $\mathbf{e}_n$ may be considered as a typical sample of $\mathbf{E}_n$.

\begin{definition}[Random network ensemble]\rm
\label{def:random}
Let $\mathcal{G}_n(\mathbf{e}_n)$ be the set of all  weighted directed graphs with degree
sequence $\mathbf{d}^+_n,\bfdminusn$ such that, for any node $i$, the set of exposures is given by the non-zero elements of line $i$ in the exposure matrix $\mathbf{e}_n$.
Let  $(\Omega, \cal{A}, \mathbb{P})$ be a probability space. We define $\mathbf{E}_n:\Omega\to \mathcal{G}_n(\mathbf{\en})$ as a random directed graph uniformly distributed on $\mathcal{G}_n(\mathbf{e}_n)$.
\end{definition}
We endow the nodes in $\mathbf{E}_n$ with the capital ratios
$\mathbf{\gamma}_n$. Then for all $i=1, \dots, n,$
$$\quad\{ \En(i,j),\quad \En(i,j)\neq 0 \}=  \{ \en(i,j),\quad \en(i,j)\neq 0 \}\quad \mathbb{P}-a.s.$$
$$\#\{j\in V,\  \En(j,i) > 0\}=d^+_n(j),\quad{\rm and}\qquad \#\{j\in V,\  \En(i,j)\neq 0 \}=d^-_n(i).$$

Definition \ref{def:random} is equivalent to the
representation of the financial system by an unweighted graph
chosen uniformly among all graphs with the degree sequence
$(\bfdplusn, \bfdminusn)$, in which we assign to the links emanating from node $i$
 the set of weights $\{e_n(i, j) > 0\}$.



\section{Asymptotic results}
\label{sec:main} We consider a sequence of random financial networks as
introduced above. Our goal is to study the behavior of
$\alpha_n(\mathbf{E}_n, \mathbf{\gamma}_n)$ which represents the
final fraction of defaults in the cascade generated by the set of initially
insolvent institutions, i.e., $\mathbb{D}_0(\mathbf{E}_n,
\mathbf{\gamma}_n)=\{ i \in [n] \mid \ \gamman(i)=0\}$.

\subsection{Some probability-theoretic notation}
We let $\NN_0$ be the set of non-negative integers, i.e., $\NN_0 = \NN \cup \{0\}$. For non-negative sequences $x_n$ and $y_n$, we describe their relative order of magnitude using Landau's $o(.)$ and $O(.)$ notation. We
write $x_n = O(y_n)$ if there exist $N \in \mathbb{N}$ and $C > 0$ such that $x_n \leq C y_n$ for all $n \geq N$, and
$x_n = o(y_n)$, if $x_n/ y_n \rightarrow 0$, as $n \rightarrow \infty$.

Let $\{ X_n \}_{n \in \mathbb{N}}$ be a sequence of real-valued random variables on a sequence of probability spaces
$\{ (\Omega_n, \mathbb{P}_n)\}_{n \in \mathbb{N}}$.
If $c \in \mathbb{R}$ is a constant, we write $X_n \stackrel{p}{\rightarrow} c$ to denote that $X_n$ \emph{converges in probability to $c$}.
That is, for any $\epsilon >0$, we have $\mathbb{P}_n (|X_n - c|>\epsilon) \rightarrow 0$ as $n \rightarrow \infty$.

Let $\{ a_n \}_{n \in \mathbb{N}}$ be a sequence of real numbers that tends to infinity as $n \rightarrow \infty$.
We write $X_n = o_p (a_n)$, if $|X_n|/a_n$ \emph{converges to 0 in probability}.
Additionally, we write $X_n = O_p (a_n)$, to denote that for any positive-valued function $\omega (n) \rightarrow \infty$,
as $n \rightarrow \infty$, we have $\mathbb{P} (|X_n|/a_n \geq \omega (n)) = o(1)$.
If $\mathcal{E}_n$ is a measurable subset of $\Omega_n$, for any $n \in \mathbb{N}$, we say that the sequence
$\{ \mathcal{E}_n \}_{n \in \mathbb{N}}$ occurs \emph{with high probability (w.h.p.)} if $\mathbb{P} (\mathcal{E}_n) = 1-o(1)$, as
$n\rightarrow \infty$.

Also we denote by $\Bin (k,p)$ denotes a binomially distributed random variable corresponding to the number of
successes of a sequence of $k$ independent Bernoulli trials each having probability of success equal to $p$.

\subsection{Assumptions}
Consider a sequence $(\mathbf{e}_n,\gamman)_{n\geq 1}$ of financial
networks, indexed by the number of nodes $n$. Let $m_n$ denote the total number of links in the network $\mathbf{e}_n$:
$$m_n := \sum_{i=1}^n \dplusn(i) = \sum_{i=1}^n \dminusn(i).$$
The \emph{empirical distribution of the degrees} is defined by
$$\mu_n(j,k) := \frac{1}{n} \# \{i\in [n] \mid \ \dplusn(i)=j, \dminusn(i)=k\}.$$

From now on, we assume that the degree sequences $\bfdplusn$ and $\bfdminusn$ satisfy the following conditions analogous to the ones introduced in \cite{Molloy98thesize}.

\begin{assumption} \rm
\label{cond}
For each $n \in \NN$, $\bfdplusn=\{(\dplusn(i))^n_{i=1}\}$ and $\bfdminusn=\{(\dminusn(i))^n_{i=1}\}$ are sequences of non-negative integers with $\sum_{i=1}^n \dplusn(i) = \sum_{i=1}^n \dminusn(i)$, and such that for some probability distribution $\mu$ on $\mathbb{N}_0^2$ independent of $n$ and with finite mean $\lambda:= \sum_ {j,k}j \mu(j,k) = \sum_ {j,k}k \mu(j,k) \in (0,\infty)$, the following holds:
\begin{enumerate}
\item   $\mu_n(j,k) \to \mu(j,k)$ for every $j,k \geq 0$ as $n \to \infty$ ;
\item  $\sum_{i=1}^n (\dplusn(i))^2 + (\dminusn(i))^2 = O(n)$.
\end{enumerate}
\end{assumption}

\vspace{0.3cm}
Note that, in particular the second assumption implies (by uniform integrability) that $m_n /n \to \lambda$, as $n \to \infty$. 

\vspace{0.3cm}
We now present our assumptions on the exposures. Let us denote by $\Sigma_n(i)$ the set of all permutations of the counterparties of $i$ in the network $\mathbf{e}_n$, i.e., permutations of the set $\{ j\in [n] \mid \ e_n(i,j)>0 \}$. For the purpose of studying contagion, the role of exposures and capital ratios may be expressed in terms of  {\it default thresholds} for each node.
\begin{definition}[Default threshold]
\label{thresholdfunction}\rm
For a node $i$ and permutation $\tau_n \in \Sigma_n(i)$ which specifies the order in which $i$'s counterparties default,  the default threshold
\begin{equation}
\Theta_n(i, \tau_n) := \min\{k \geq 0 \mid \ {\gamma_n(i)\sum_{j = 1}^{n} e_n(i,j) < \sum_{j = 1}^{k}(1 - R) e_n(i,\tau_n(j))}\}
\end{equation}
 measures how many counterparty defaults $i$ can tolerate before it becomes insolvent (in the financial network $(\mathbf{e}_n,\gamman)$), if its counterparties default in the order specified by $\tau_n$.
\end{definition}

We also define
\begin{eqnarray}
p_n(j,k,\theta) := \frac{\#\{(i,\tau_n) \mid \ i \in [n], \ \tau_n \in \Sigma_n(i),
 \ \dplusn(i) = j, \ \dminusn(i) = k, \ \Theta_n(i, \tau_n) = \theta\}}{n \mu_n(j,k) j!}.
\end{eqnarray}
 We will see in Section~\ref{sec-coupling} that for $n$ large, $p_n(j,k,\theta)$ gives the fraction of nodes with degree $(j,k)$ which have a default threshold equal to $\theta$, in the random financial network $\mathbf{E}_n$.

 In particular for $\theta = 1$, $$n\mu_n(j,k)jp_n(j,k,1)$$ is  the  number of exposures of nodes with degree $(j,k)$ which exceed the capital of the exposed node, i.e., exposures which in case of default of the initial node always lead to the insolvency of the exposed node. These links play a crucial role (as we will see in Section \ref{sec:resilience}) and we call them \emph{contagious links}.
\begin{definition}[Contagious link]\rm
\label{contagiousLinks}
We call a link $i \to j$  {\it contagious} if it represents an exposure larger than the capital of the exposed node:
$$ (1-R) e_n(i, j) > c_n(i) = \gamma_n(i) \sum_{j=1}^n e_n(i,j).$$
\end{definition}

From now on, we assume that  $p_n(j,k,\theta)$ has a limit  when $n \rightarrow \infty$.

\begin{assumption}\label{condition_threshold} \rm
There exists a function $p: \mathbb{N}_0^{3} \to [0,1]$ such that for all $j,k , \theta \in \NN_0$ ($\theta \leq j$)
\begin{equation*}
p_n(j,k,\theta) \to p(j,k,\theta), \ \mbox{as} \ n \to \infty.
\end{equation*}
\end{assumption}

Some examples of exposures for which this assumption is fulfilled are given in Section \ref{sse-examples}.
 Under this assumption, we will see in Appendix \ref{sec:proofs} that $p(j,k,\theta)$ is also the limit in probability of the fraction of nodes with degree $(j,k)$ which become insolvent after $\theta$ of their counterparties default.  In particular,
 \begin{itemize}
\item $p(j,k,0)$ represents the proportion of initially insolvent nodes with degree $(j,k)$;
\item $p(j,k,1)$ represents the proportion of nodes with degree $(j,k)$ which are  `vulnerable', i.e., may become insolvent due to the default of a single counterparty.
\end{itemize}

\subsection{The asymptotic magnitude of contagion}
Consider a sequence $(\mathbf{e}_n,\gamman)_{n\geq 1}$ of financial networks satisfying Assumptions \ref{cond} and \ref{condition_threshold}, and let $(\mathbf{E}_n,\gamman)_{n\geq 1}$ be their corresponding sequence of random financial networks, see Definition \ref{def:random}. Let us denote by
\begin{eqnarray*}
\beta(j,\pi,\theta) := \PP(\Bin(j,\pi) \geq \theta) = \sum_{l \geq \theta}^j{j \choose l} \pi^l (1-\pi)^{j-l},
\end{eqnarray*}
the distribution function of a binomial random variable  $\Bin(j,\pi)$  with parameters $j$ and $\pi$.


We define the function $I:[0,1] \to [0,1]$ as
\begin{equation}
I(\pi):= \sum_{j, k} \frac{\mu(j, k) k}{\lambda}\sum_{\theta = 0}^{j} p(j, k, \theta)\beta(j,\pi,\theta).
\end{equation}
Indeed, $I(\pi)$ has the following interpretation (when the network size goes to infinity): If the
end node of a randomly chosen edge defaults with probability
$\pi$, $I(\pi)$ is the expected fraction of counterparty
defaults after one iteration of the cascade.

Let $\pi^*$ be the smallest fixed point of $I$ in $[0,1]$, i.e.,
$$\pi^* = \inf \{\pi\in[0,1] \mid \ I(\pi) = \pi\} .$$
(The value $\pi^*$ represents the probability that an
edge taken at random ends in a defaulted node, at the end of the contagion process.)

\begin{remark}\rm
$I$ admits at least one fixed point. Indeed, $I$ is a continuous increasing
function and,
$$I(1) = \sum_{j, k} \frac{\mu(j, k) k}{\lambda} \sum_{\theta = 0}^{j} p(j, k, \theta) \leq 1$$
since $\sum_{\theta} p(j, k, \theta) \leq 1$ by definition.
Moreover,
$$I(0) = \sum_{j, k} \frac{\mu(j, k) k}{\lambda}  p(j, k, 0) \geq 0.$$
So the function $I$ has at least a fixed point in $[0,1]$.
\end{remark}

We can now announce our main theorem.
\begin{theorem}
\label{thm-main}
Consider a sequence $(\mathbf{e}_n,\gamman)_{n\geq 1}$ of financial networks satisfying Assumptions \ref{cond} and \ref{condition_threshold}, and the corresponding sequence of random matrices $(\mathbf{E}_{n})_{n \geq 1}$ defined on $(\Omega, \cal{A}, \mathbb{P})$ as in Definition \ref{def:random}.
Let $\pi^*$ be the smallest fixed point of $I$ in $[0,1]$.
\begin{enumerate}
\item If $\pi^*=1$, i.e., if $I(\pi) > \pi$ for all $\pi \in [0,1)$, then asymptotically almost all nodes default during the cascades
$$\alpha_n( \mathbf{E}_n, \mathbf{\gamma}_n) \stackrel{p}{\rightarrow} 1.$$
\item If $\pi^*<1$ and furthermore $\pi^*$ is a stable fixed point of $I$, i.e., $I'(\pi^*) < 1$, then the asymptotic fraction of defaults is given by
$$\alpha_n( \mathbf{E}_n, \mathbf{\gamma}_n) \stackrel{p}{\rightarrow} \sum_{j, k} \mu(j, k)\sum_{\theta = 0}^{j}p(j, k, \theta) \beta(j, \pi^*, \theta).$$
\end{enumerate}
\end{theorem}
A proof of this theorem is given in Appendix \ref{sec:proofs}.

\subsection{Resilience to contagion}\label{sec:resilience}

The  resilience of a network to small shocks is a global property of the network which depends on its detailed structure. However, the above results allow us to introduce a rather simple and easy to compute indicator for the resilience of a network to small shocks. Consider a sequence $(\mathbf{e}_n,\gamman)_{n\geq 1}$ of financial networks satisfying Assumptions \ref{cond} and \ref{condition_threshold}.
\begin{definition}[Network resilience]\rm
\label{resilience.def}
We define the \emph{network resilience} function as
$$1 - \sum_{j,k}  \frac{jk}{\lambda} \mu(j,k) p(j,k,1)\in (-\infty, 1].$$
\end{definition}
The following result, which is a consequence of Theorem \ref{thm-main}, shows that this indicator measures the resilience of a network to   the initial default
of  a small fraction $\epsilon$ of the nodes:
\begin{proposition}\label{thm-resilient}
Consider a sequence $(\mathbf{e}_n,\gamman)_{n\geq 1}$ of financial networks satisfying Assumptions \ref{cond} and \ref{condition_threshold}, and let $(\mathbf{E}_n,\gamman)_{n\geq 1}$ be their corresponding sequence of random financial networks, see Definition \ref{def:random}. If
\begin{equation}
1 - \sum_{j,k}  \frac{jk}{\lambda} \mu(j,k) p(j,k,1) > 0,
\label{cond-resilient}
\end{equation}
then for every $\epsilon > 0$ there exists $N_{\epsilon}$ and $\rho_{\epsilon}$ such that, if the initial fraction of defaults is smaller than  $\rho_{\epsilon}$ then the final fraction of defaults is negligible with high probability: $$\forall n \geq N_{\epsilon},\qquad \PP(\alpha_n(\mathbf{E}_n, \mathbf{\gamma}_n) \leq \epsilon) > 1 - \epsilon.$$
\end{proposition}
\begin{proof}
Consider $\rho$ be defined as the fraction of fundamental defaults
$$\rho := \sum_{j, k}\mu(j, k)p(j,k,0).$$
We have
$$ I(\alpha) = \sum_{j, k}\frac{\mu(j, k) k}{\lambda} \sum_{\theta = 0}^{j} p(j, k, \theta) \beta(j, \alpha, \theta).$$
Using a first order expansion of $\beta(j, \alpha, \theta)$ in
$\alpha$ at $0$ (when $\alpha \to 0$), we obtain
$$\beta(j, \alpha, \theta) = 1_{\{\theta = 0\}} + \alpha j 1_{\{\theta = 1\}} + o(\alpha) .$$
Thus,
$$ I(\alpha) = \sum_{j, k}\frac{\mu(j, k) k}{\lambda}(p(j,k,0) + \alpha j p(j, k, 1)) + o(\alpha).$$
Let $\alpha^*$ be the smallest fixed point of $I(\alpha)$.
Given Condition (\ref{cond-resilient}), for $\alpha > 0$ and small
enough, we have
$$ \mathop{\lim}_{\rho \to 0}I(\alpha) = \alpha\sum_{j, k}\frac{\mu(j, k) jk}{\lambda}  p(j, k, 1) + o(\alpha) < \alpha,$$
where we used the fact that if the fraction of fundamental defaults tends to zero, so does the fraction of out-going links belonging to fundamentally defaulted nodes.
On the other hand we have seen that $ I(0) \geq 0 $.
Thus $\mathop{\lim}_{\rho \to 0} \alpha^* = 0$.

Let us now fix  $\epsilon > 0$.
By  continuity of the function $g$ defined by $$g(\alpha) := \sum_{j, k}\mu(j, k) \sum_{\theta = 0}^{j} p(j, k, \theta) \beta(j, \alpha, \theta),$$ appearing in Theorem \ref{thm-main}, there exists $\rho_{\epsilon}$ such that
$g(\alpha^*) < \epsilon/2$ as soon as $\rho < \rho_{\epsilon}$.
By Theorem \ref{thm-main} we have that there exists an integer $N_{\epsilon}$ such that, for $n \geq N_{\epsilon}$,
$$ \PP(|\alpha_n(\mathbf{E}_n, \mathbf{\gamma}_n) - g(\alpha^*)| < \epsilon/2) > 1 - \epsilon, $$
which completes the proof.
\end{proof}

The proof points out what is a natural and intuitive fact: If contagion does not spread to nodes with contagion threshold $1$, then it will not spread at all.

\begin{theorem}
\label{thm-giant}
Consider a sequence $(\mathbf{e}_n,\gamman)_{n\geq 1}$ of financial networks satisfying Assumptions \ref{cond} and \ref{condition_threshold}, and let $(\mathbf{E}_n,\gamman)_{n\geq 1}$ be their corresponding sequence of random financial networks. If
\begin{equation}
1 - \sum_{j,k}  \frac{\mu(j, k) jk}{\lambda} p(j,k,1) < 0,
\label{cond-giant}
\end{equation}
then with high probability there exists a set of nodes representing a
positive fraction of the financial system, strongly interlinked by contagious links (i.e., there is a directed path of contagious links from any node to another in the component), such
that any node belonging to this set can trigger the default of all
nodes in the set.
\end{theorem}

Given the network topology, Condition (\ref{cond-resilient}) sets
limits on the  fraction of contagious links
$p_n(j,k,1)$, i.e., on the magnitude of exposures relative to capital.
A proof of this theorem is given in Appendix \ref{sec:proofs}. However, Condition \ref{cond-resilient} may be justified using the following heuristic argument.

\begin{remark}[Branching process approximation] \rm We describe an approximation of the local structure of the graph by a branching process, the children being the in-coming neighbors: the root $\phi$ with probability $\mu ^-(k_{\phi}) := \sum_j \mu(j,k_{\phi})$ has an in-degree equal to $k_{\phi}$. Each of these $k_{\phi}$ vertices with probability $\frac{\mu(j, k) j}{\lambda}$ has degree $(j,k)$, and with probability equal to $p(j,k,1)$ default when their parent defaults. Let $y$ be the extinction probability, given by the smallest solution of
\begin{eqnarray}
\label{eq-bp}
y = \sum_{j,k} \frac{\mu(j, k) j}{\lambda} p(j,k,1) y^k .
\end{eqnarray}
If $\sum_{j,k} \frac{\mu(j, k) jk}{\lambda} p(j,k,1) < 1,$ then the smallest solution of (\ref{eq-bp}) is $y=1$ (the population dies out with probability one), whereas if $$\sum_{j,k} \frac{\mu(j, k) jk}{\lambda} p(j,k,1) > 1,$$ there is a unique solution with $y \in (0,1)$.
\end{remark}

\begin{remark}[Too interconnected to fail?]\rm
We suppose that the resilience condition (\ref{cond-resilient}) is
satisfied. Let $\pi^*_{\epsilon}$ be the smallest fixed point of $I$
in $[0,1]$, when a fraction $\epsilon$ of all nodes represent
fundamental defaults, i.e., $p(j,k,0) = \epsilon $ for all $j,k$.

We obtain then, by a first order approximation of the function $I$,
that
$$\pi^*_{\epsilon} = \frac{\epsilon}{1- \sum_{j, k} \frac{\mu(j, k) jk}{\lambda}p(j, k, 1)} + o(\epsilon).$$

By a first order approximation of the function $\pi \to \sum_{j, k} \mu(j, k)\sum_{\theta = 0}^{j}p(j, k, \theta) \beta(j, \pi, \theta)$ giving
the asymptotic fraction of defaults in Theorem \ref{thm-main}, we obtain that,
for any $\rho$ there exists $ \epsilon_{\rho}$ and $n_{\rho}$ such that for all $\epsilon < \epsilon_{\rho}$ and $n > n_{\rho}$
\begin{equation}
\PP(|\alpha_n(\mathbf{\En}, \mathbf{\gamman}) - \epsilon(1 + \frac{ \sum_{j, k} j \mu(j, k) p(j, k, 1)}{1 - \sum_{j, k} \frac{\mu(j, k) jk}{\lambda}p(j, k, 1)})| < \rho) > 1 - \rho.
\label{eqn:amplification}
\end{equation}

Suppose now that  initially insolvent fraction involves only nodes with degree $(d^+,d^-)$, and
we denote $\pi^*_{\epsilon}(d^+,d^-)$ the smallest fixed point of $I$ in $[0,1]$ in the case where $p(d^+,d^-,0) = \epsilon $ and $p(j,k,0) = 0$ for all $(j,k) \neq (d^+,d^-)$.
Then we obtain that,
for any $\rho$ there exists $ \epsilon_{\rho}$ and $n_{\rho}$ such that for all $\epsilon < \epsilon_{\rho}$ and $n > n_{\rho}$,
\begin{equation}\PP\left(|\alpha_n(\mathbf{\En}, \mathbf{\gamman}) - \epsilon\mu(d^+, d^-)(1 + \frac{d^-}{\lambda}\frac{\sum_{j, k}\frac{\mu(j, k) jk}{\lambda} p(j, k, 1)}{1 - \sum_{j, k} \frac{\mu(j, k) jk}{\lambda}p(j, k, 1)})| < \rho\right) > 1 - \rho.
\label{eqn:amplification1}
\end{equation}

This simple expression shows that there are basically two factors that
determine how small initial shocks are amplified by the financial
network: the interconnectedness of the initial default --represented by its
in-degree $d^-$-- and the susceptibility of the network, given by the factor $\sum_{j, k}\frac{\mu(j, k) jk}{\lambda} p(j, k, 1)$.
\end{remark}

\section{Numerical results on finite networks}
The results of Section \ref{sec:main} hold in the  limit of large network size.  In order to assess whether these results still hold for  networks whose size is large but finite, we now compare our theoretical results with numerical simulations for networks with realistic sizes. In particular, we investigate
 the effect of heterogeneity in network structure and the relation between resilience and connectivity.
We begin by presenting some examples of models for counterparty networks which satisfy Assumption~\ref{condition_threshold}.
\label{sec:results}
\subsection{Examples of networks  which satisfy Assumption~\ref{condition_threshold}}\label{sse-examples}
In this section we give two important examples of exposures which satisfy Assumption~\ref{condition_threshold}.
\begin{example}[Independent exposures]\rm
Assume that for all $n$, the exposures of all nodes $i \in [n]$ with the same degree $(j,k)$, i.e., $$\left\{e_n(i,l) > 0 \mid \dplusn(i) = j, \dminusn(i) = k \right\},$$  are independent and identically distributed (i.i.d.) random variables, with a law given by $F_X(j,k)$, depending  on $j$ and $k$ but independent of $n$. We assume the same for the sequence of capital ratios,  i.e., $\left\{\gamma_n(i) \mid {\dplusn(i) = j, \dminusn(i) = k} \right\}$ are i.i.d. random variables with a law given by $F_{\gamma}(j, k)$ which may depend on $(j,k)$, but not on $n$. Then it is easy to see that, by the law of large numbers, Assumption \ref{condition_threshold} holds and
the limit $p(j,k,\theta)$ is known (for all $j,k,\theta$),
$$p(j,k,\theta) = \mathbb{P}(X_{\theta} > \gamma\sum_{l = 1}^j X_l - \sum_{l = 1}^{\theta - 1}(1 - R) X_l \geq 0),$$
where $\gamma$ is a random variable with law $F_{\gamma}(j, k)$, and $(X_l)_{l = 1}^j$ are i.i.d. random variables with law $F_X(j,k)$ and independent of $\gamma$.
\end{example}

\begin{example} [Exchangeable exposures]\rm
 Empirical observations of banking networks \cite{cont10c, boss04, Soramaki07} show that they   are hierarchical, `disassortative' networks \cite{May10}, with a few large and highly interconnected dealer banks and many small banks, connected predominantly to   dealer banks.
This can be modeled in a stylized way by partitioning the set of nodes into two sets, a set $\mathcal{D}$ of $n^D$ dealer banks,  and a set $\mathcal{N}$ of $n^N$ non-dealer banks.

We assume that the exposures $\{e_n(i, l) > 0 \mid i \in \mathcal{D}\}$, and $\{e_n(i, l) > 0 \mid i \in \mathcal{N}\}$ are restrictions corresponding to the first  $m_n^D$ (respectively $m_n^N$) elements of infinite sequences of exchangeable variables, where $m_n^D$ and $m_n^N$ denote the total number of exposures belonging to dealer and respectively non-dealer banks. Similarly, the capital ratios $\{\gamma_n(i) \mid i \in \mathcal{D}\}$ and $\{\gamma_n(i) \mid i \in \mathcal{N}\}$ are restrictions to the first $n^D$ (respectively $n^N$) elements of the sequence, independent of the sequence of exposures.

We can extend this example to a finite number of classes of nodes represented by their degrees, and also drop the assumption of independence between exposures and capital ratios. We assume that within each class, the sequence of a node's exposures and capital ratios are exchangeable random variables.

For each node $i$ with $\dplusn(i) = j, \dminusn(i) = k$, we let
$$Y_n(i) := \left(\{e_{n}(i,j) > 0\}, \gamma_n(i)\right)$$
be a multivariate random variable with state space $\Im^{j,k} \subset \mathbb{R}_+^{j}\otimes\mathbb{R}$. We assume that the law of the finite sequence
$$ \left\{ Y_n(i) \mid \ i \in [n], \dplusn(i) = j, \dminusn(i) = k\right\} $$ is invariant under permutation.

Then the family $\left\{ Y_n(i) \mid \ i \in [n], \dplusn(i) = j, \dminusn(i) = k\right\}_{0\leq j,k\leq M} $ represents a family of finite multi-exchangeable systems, as defined by Graham \cite{Graham08}. It is proved in \cite{Graham08} that the conditional law of a finite multiclass system, given the value of the vector of the empirical measures of its classes, corresponds to independent uniform orderings of the samples within each class, and that a family of such systems converges in law if and only if the corresponding empirical measure vectors converge in law.

Let us consider the empirical measure sequence
$$ \left\{\Lambda_n^{j, k} := \frac{\sum_i \ind_{\{\ \dplusn(i) = j, \ \dminusn(i) = k\}} \delta_{Y_n(i)}}{n \mu_n(j,k)}\right\}_{0\leq j,k\leq M} .$$

We suppose that the family $\left\{ Y_n(i) \mid i \in [n], \dplusn(i) = j, \dminusn(i) = k\right\}_{0\leq j,k\leq M}$ converges in law when $n\to\infty$ to an infinite multi-exchangeable system
\begin{equation}
 \lim_{n \to \infty}\left\{Y_n(i) \mid i \in [n], \dplusn(i) = j, \dminusn(i) = k\right\}_{0\leq j,k\leq M} \stackrel{\mathcal{L}}{=}\left\{Z_l^{j,k}\mid l \geq 1\right\}_{0\leq j,k\leq M}.
\label{infexch}
\end{equation}
By  \cite[Theorem 2]{Graham08}, the empirical measure
converges in law to
\begin{equation}
\lim_{n \to \infty}\left\{\Lambda_n^{j, k}\right\}_{0\leq j,k\leq M} \stackrel{\mathcal{L}}{=}\left\{\Lambda^{j, k}\right\}_{0\leq j,k\leq M}.
\label{measureconv}
\end{equation}
For an arbitrary $Z \in \Im^{j,k}$, we define the function
$$h(Z, \theta) = \frac{\#\{\tau \mid  \tau \in \Sigma(j), \Theta(Z,\tau) = \theta\}}{j!}.$$

Thus, by Equation (\ref{measureconv}) giving the convergence of empirical measures and, the fact that the function $h$
is bounded, we have
$$p_n(j,k,\theta) = \mathbb{E}^{\Lambda_n^{j,
k}}(h(\mathbf{Z}, \theta)) \stackrel{n \to \infty}{\to}
\mathbb{E}^{\Lambda^{j, k}}(h(\mathbf{Z}, \theta)) =
p(j,k,\theta),$$
with $\mathbf{Z}$ a random element of $\Im^{j,k}$ and $\mathbb{E}^{\Lambda_n^{j,
k}}$ and $\mathbb{E}^{\Lambda^{j, k}}$ denoting expectation under the measures $\Lambda_n^{j,
k}$ and $\Lambda^{j, k}$ respectively.
A last observation is that Equation (\ref{infexch}) is verified in our two tiered example since the sequences of exposures in the network of size $n$ are  restrictions of infinite exchangeable sequences.
\end{example}

\subsection{Relevance of asymptotics}
Interbank networks  in developed countries may contain several
thousands of nodes. The Federal Deposit Insurance Corporation
insured  $7969$ institutions as of $3/18/2010$,
while the European Central Bank reports   $8350$ monetary
financial institutions in the Euro zone ($80 \%$ credit institutions
and $20\%$ money market funds). To assess the
relevance of asymptotic formulae for studying contagion in networks with such sizes,  we  generate
a scale-free network of $10000$ nodes with Pareto distributed exposures using
the random graph model introduced by Blanchard \cite{blanchard03}, which can be seen as a static version
of the preferential attachment model. In this model, given  the sequence of out-degrees, an arbitrary out-going edge is  assigned to an end-node $i$ with probability proportional to the power $d_n^+(i)^\alpha$ where  $\alpha>0$. This leads to positive correlation between in-degrees and out-degrees.

The   distribution of the out-degree in this model is a Pareto law with tail exponent $\gamma^+$:
$$\mu_n^+(j) := \#\{i \mid d_n^+(i) = j\} \stackrel{n\to\infty}{\rightarrow} \mu^+(j) \sim j^{\gamma^+ + 1},$$
and the conditional limit law of the in-degree is a Poisson distribution
$$ P(d^- = k|d^+ = j) = e^{-\lambda(j)}\frac{\lambda(j)^{k}}{k!},$$ with
$\lambda(j) = \frac{j^{\alpha}\EE(D^+)}{\EE((D^+)^{\alpha})},$ where $D^+$ denotes a random variable with law $\mu^+$.
The main theorem in \cite{blanchard03} states that the marginal distribution of the in-degree has a Pareto tail with exponent $\gamma^- = \frac{\gamma^+}{\alpha}$, provided $1 \leq \alpha < \gamma^+$.

\begin{center}\it
In the sequel we investigate cascades on these finite networks. A crucial ingredient is the capital ratio of each bank. Throughout this section, the capital ratio is assumed to be bounded from below by a minimal capital ratio and, we consider the worst case scenario where all nodes have a capital ratio equal to the minimal capital ratio:
$ \gamma(i) = \gamma_{min}, \ \forall i \in V.$
\end{center}

The distribution of this
simulated network's degrees and exposures is given in Figure
\ref{fig:blanchard} and is based on the empirical analysis of the
Brazilian network in June 2007 \cite{cont10c}.

\begin{figure*}[htp]
\centering
\includegraphics[width=0.4\textwidth]{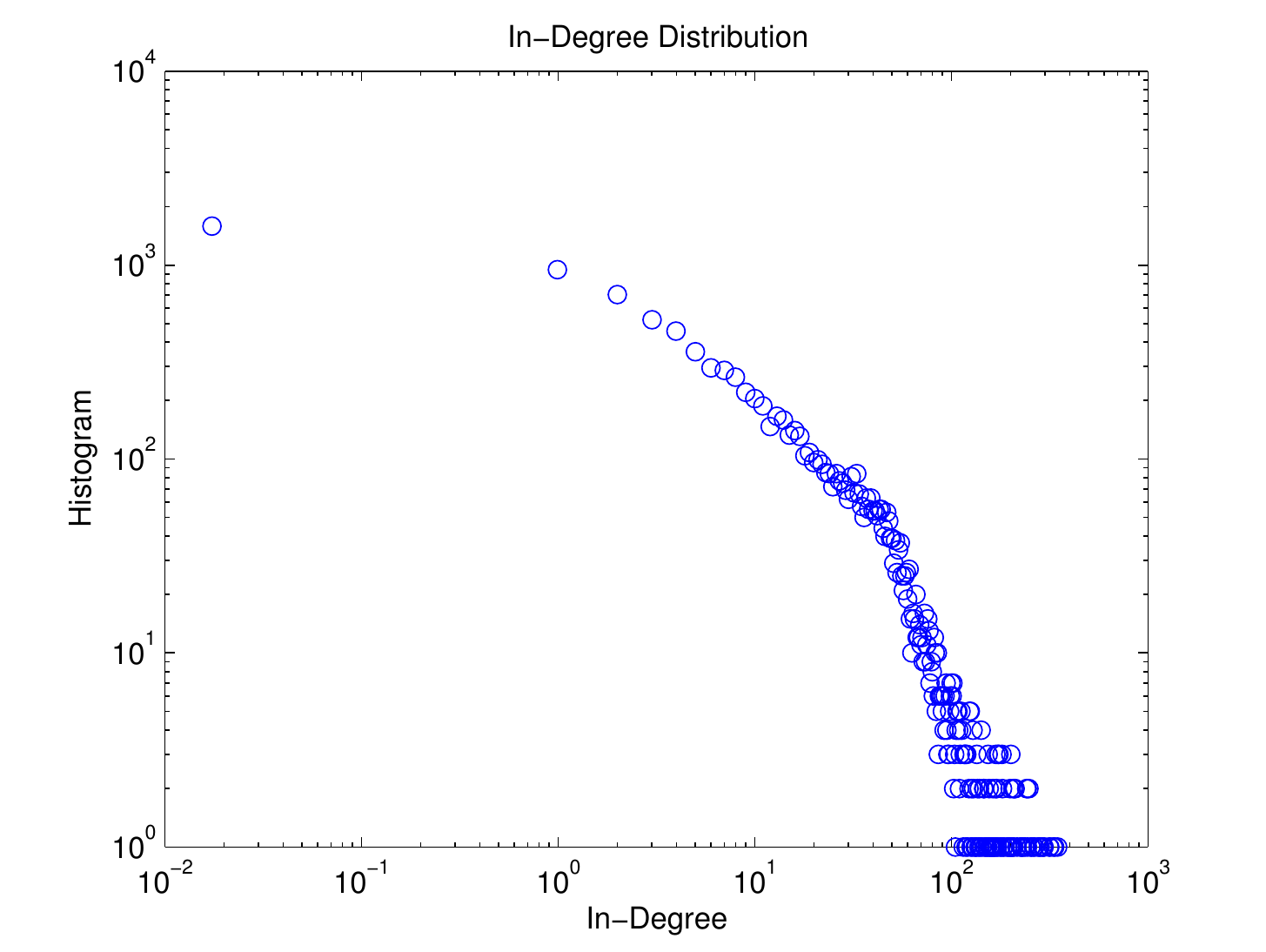}
\includegraphics[width=0.4\textwidth]{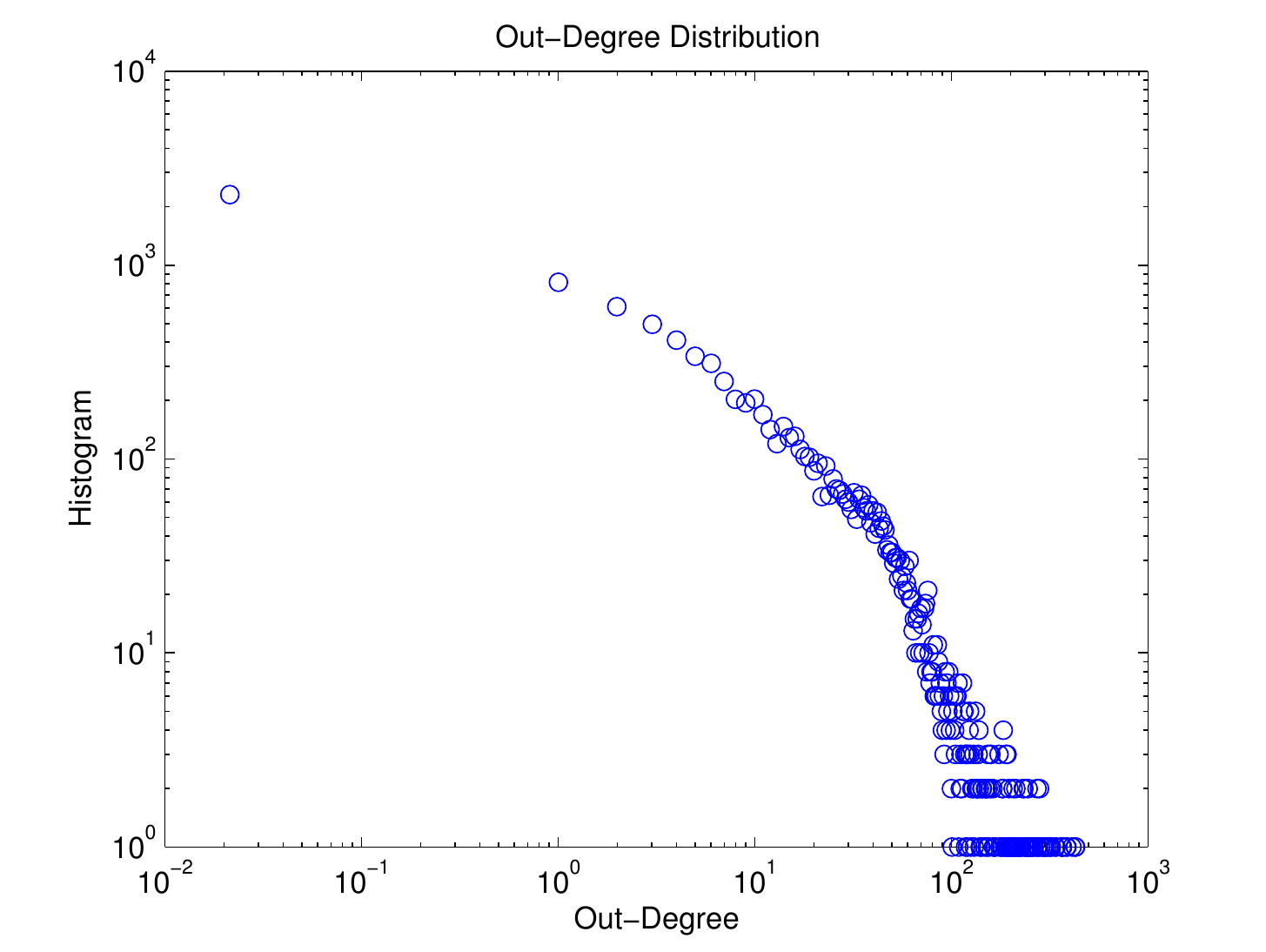}
\includegraphics[width=0.4\textwidth]{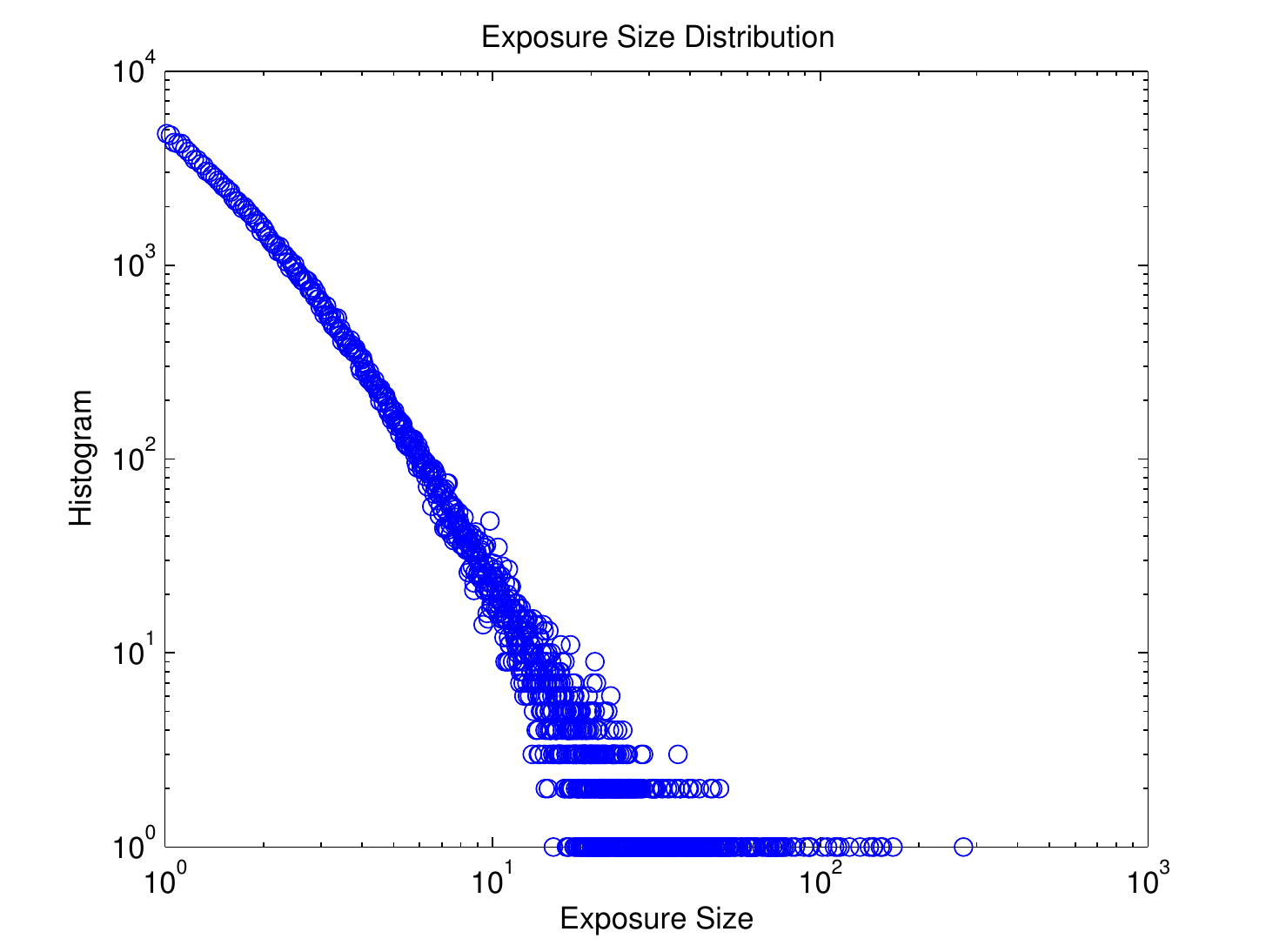}
\caption{ (a) The distribution of out-degree has a Pareto tail
with exponent $2.19$, (b) The distribution of the in-degree  has a Pareto tail  with exponent $1.98$, (c) The distribution of the
exposures (tail-exponent $2.61$).} \label{fig:blanchard}
\end{figure*}

On one hand we make a simulation of the default contagion starting
with a random set of defaults representing $0.1\%$ of all nodes
(chosen uniformly among all nodes). On the other hand we plug the
empirical distribution of the degrees and the fraction of contagious
links into Equation (\ref{eqn:amplification}) for the amplification of
a  small number of initial defaults. Figure \ref{fig:amplifSF}
plots the amplification for varying values of the minimal capital ratios.
\begin{figure*}[htp]
\centering
\includegraphics[width=0.8\textwidth]{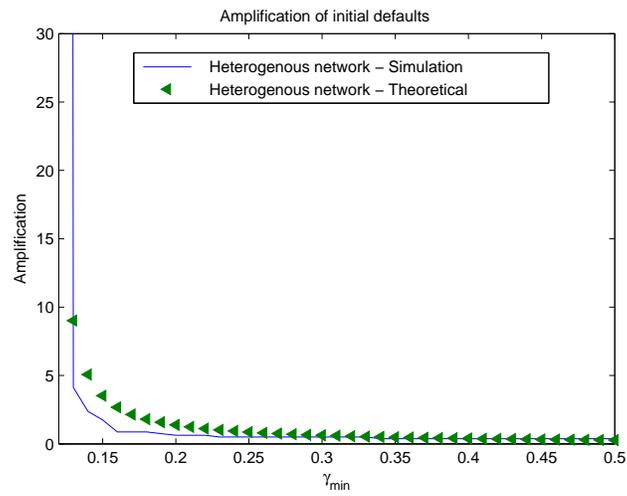}
\caption{Amplification of the number of defaults in a Scale-Free Network. The in- and out-degree of the scale-free network are Pareto distributed with tail coefficients $2.19$ and $1.98$ respectively, the exposures are Pareto distributed with tail coefficient $2.61$, $n = 10000$.}
\label{fig:amplifSF}
\end{figure*}
We find a good agreement between the theoretical and the simulated
default amplification ratios.
 We can clearly see that for
minimal capital ratios $\gamma_{min}$ less than  the critical value
$\gamma_{min}^*$, the  amplification ratio increases dramatically.

 Figure
\ref{DI} plots the simulated  fraction of defaults in a scale free network, starting
from the initial default of a single node, as a
function of the in-degree of the defaulting node, versus the theoretical slope given
in Equation (\ref{eqn:amplification1}).
Recall that the theoretical slope $\sum_{j, k}\frac{\mu(j, k) jk}{\lambda} p(j, k, 1)$ in Equation (\ref{eqn:amplification1})  measures the susceptibility of the network to initial defaults.

The agreement is particularly good for nodes with large in-degrees, pointing to large amplification if such nodes were to default in a network with high susceptibility.

\begin{figure*}[htp]
\centering
\includegraphics[width=0.6\textwidth]{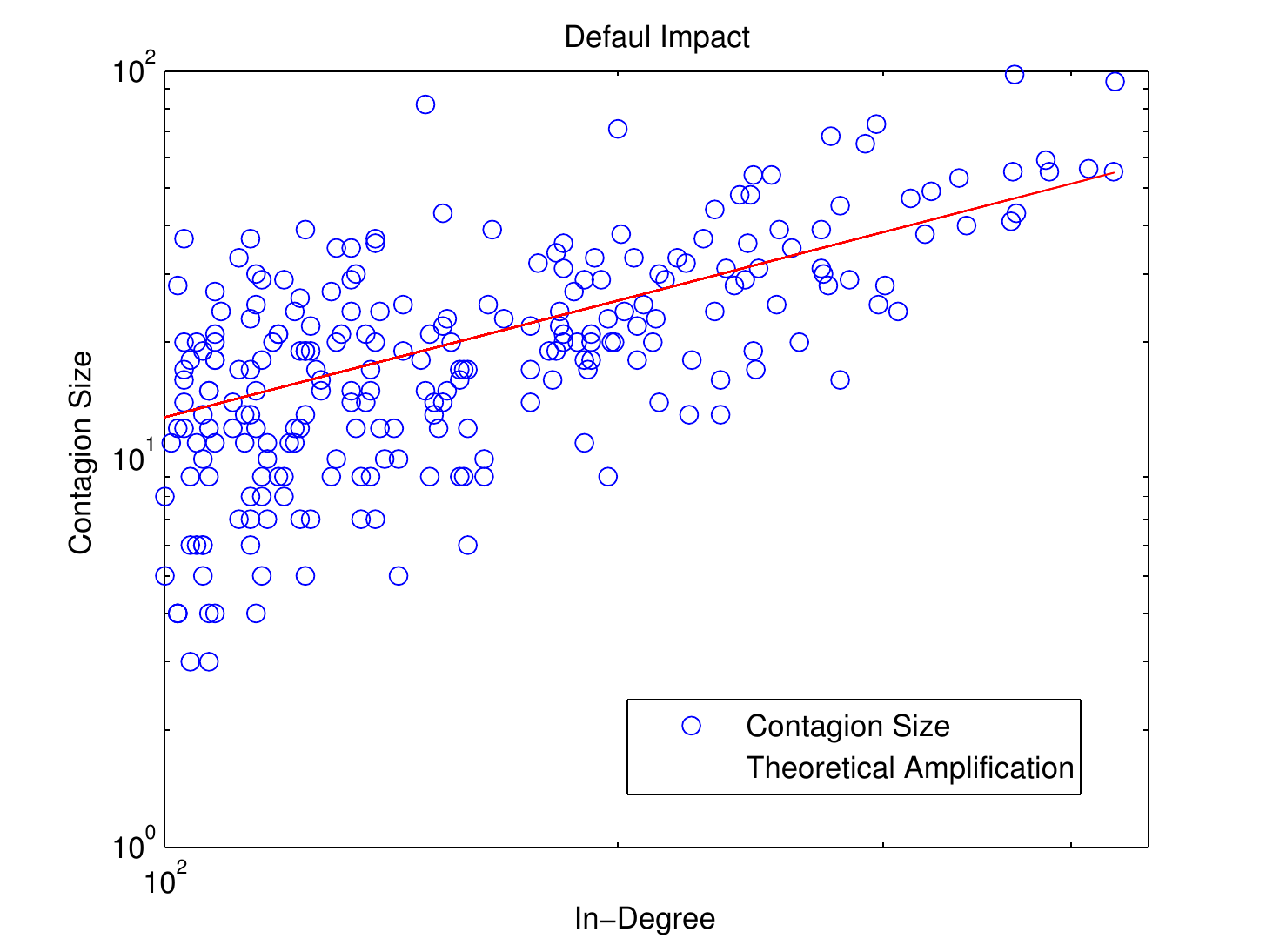}
\caption{Number of defaulted nodes} \label{DI}
\end{figure*}
\subsection{The impact of heterogeneity}
\label{sec:hetero}
In the previous examples we can compute the
minimal capital ratio $\gamma_{min}^*$ above which the network is resilient under Condition \eqref{cond-resilient}.
Two factors contribute to the sum in Condition \eqref{cond-resilient}: connectivity of the node, and
its fraction of contagious links. We compare in Figure \ref{fig:amplif}, the ratio by which contagion amplifies the number of initial default in three cases: a
scale free network with heterogeneous weights (exposures), a scale free network
with equal weights and a `homogeneous' random network (the Erd\H{o}s-R\'enyi random graph where every directed edge is present with a fixed probability) with equal weights. All three networks are parameterized to have the same average degree, i.e., the same total number of links.
\begin{figure}[htp]
\centering
\includegraphics[width=0.7\textwidth]{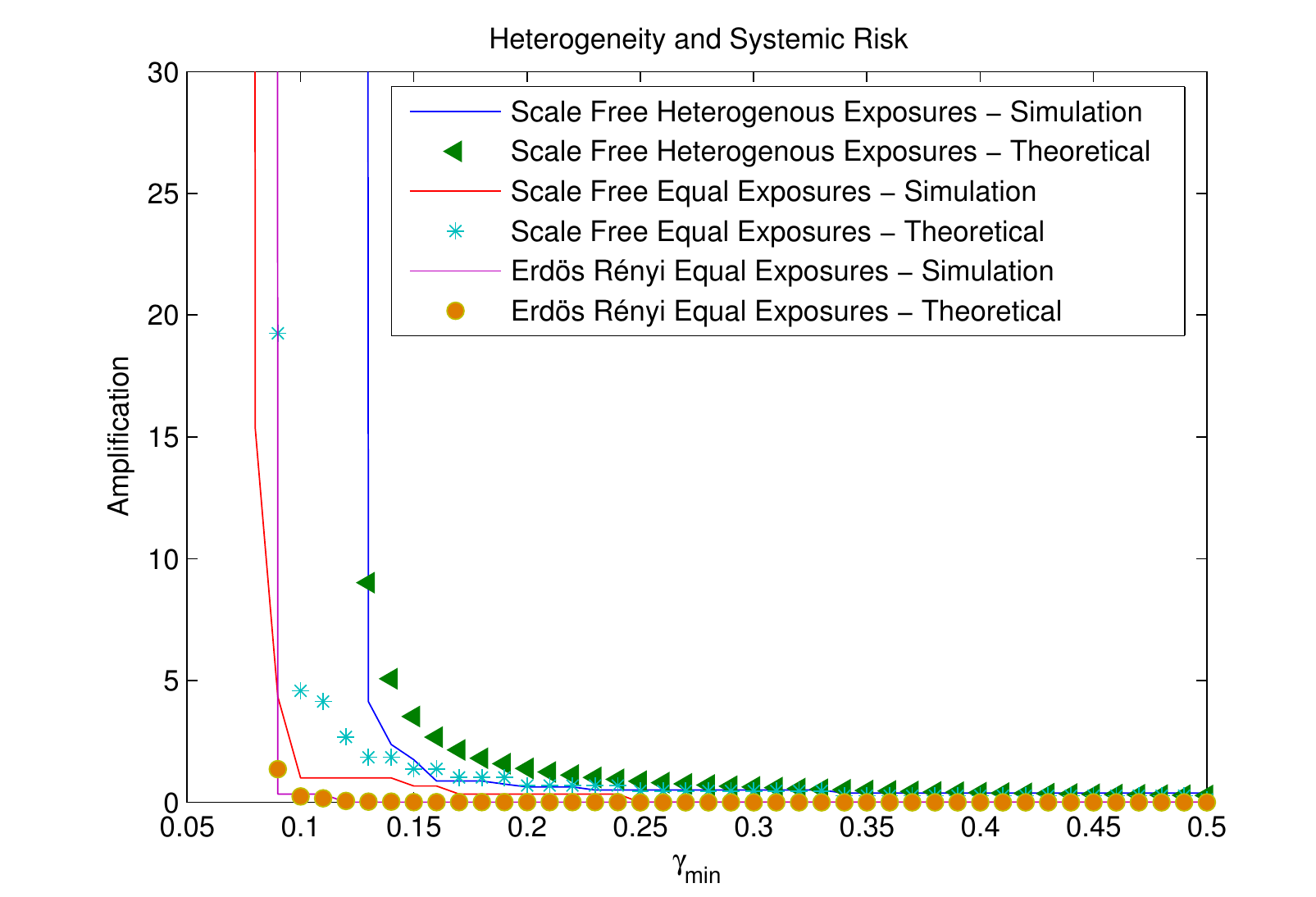}
\caption{Amplification of the number of defaults in a Scale-Free Network (in- and out-degree of the scale-free network are Pareto distributed with tail coefficients $2.19$ and $1.98$ respectively, the exposures are Pareto distributed with tail coefficient $2.61$), the same network with equal weights and an Erd\"{o}s R\'enyi Network with equal exposures $n = 10000$.}
\label{fig:amplif}
\end{figure}
It is interesting to note that in our example the most heterogeneous network is also the least resilient, as opposed
to the homogeneous Erd\"{o}s--R\'enyi network with the same distribution of
exposures.
These simulations corroborate the role played both by the network topology and the heterogeneity of weights.
\subsection{Average connectivity and contagion}
One of the recurrent questions in the economics literature regards the impact of connectivity on resilience to contagion: Is increased global connectivity posing a threat to financial system or does it allow for more efficient risk sharing?
While an influential paper by Allen and Gale \cite{allen00} finds that resilience
increases with connectivity, Battiston et al. \cite{Battiston09}  exhibit different model settings where this relation is
non-monotonous.
Our results show nonetheless that in a model with heterogenous  exposures, the average connectivity is too simple summary statistic of the topology to be explain resilience.
To see this let us consider a simple example and use the asymptotic formula (\ref{eqn:amplification}) for the amplification.

For simplicity, consider networks in which nodes' exposures are equal and $1/3 \leq \gamma_{min}< 1/2$ such
that $p_n(j, k, \theta) = 1_{\{j = 1, 2\}}$. We consider three cases of degree distributions.

First, let  $\mu_n(1, 3) = \mu_n(2, 3) = \mu_n(4, 3) = \mu_n(5, 3) =
1/4$. The average connectivity in a network with this degree distribution is $3$ and the resilience measure is equal to $1/4$.

Second,  let $\tilde{\mu}_n(1, 2) = 2/3, \tilde{\mu}_n(4, 2) = 1/3$.
The average connectivity is $2$ and the resilience measure is equal to $1/3$.

Last, we take   $\hat{\mu}_n(4, 4) = 1$ i.e., a regular graph with degree 4.
Here the average connectivity is $4$ and the resilience measure is $1$.

In all three cases a network constructed with the empirical degree distribution is resilient w.h.p..
Nonetheless, we clearly  observe that the resilience measure does not depend on the average connectivity in a monotonous way.

While in the case of \cite{Battiston09} this non-monotonicity is obtained by introducing an ad-hoc mechanism of `financial accelerators' on top of the network contagion effects,   in our case it stems from an intrinsic trade-off between risk-sharing and contagion which is inherent in the model.

These examples show that the resilience of a network cannot be simply assessed by examining an aggregate measure of connectivity such as the average degree or the number of links, as sometimes naively suggested in the literature, but requires a closer examination of features such as the distribution of degrees and the structure of the subgraph of contagious links.

\section{Conclusions}
In this paper we have analyzed distress propagation in a financial system where banks strongly differ in local features, and links between financial institutions are heterogenous in nature.

We obtained an asymptotic expression for the size of a default cascade in a large network with prescribed characteristics --degree sequence and local features of the nodes-- extending previous results for homogeneous undirected random graphs to heterogeneous, weighted directed networks.

Our asymptotic results were corroborated with a simulation study of contagion on a  network with large but realistic size: {\it on a given network}, the spread of distress can be predicted by our measure of resilience.  Our sample network has the same empirical properties as a real interbank network, e.g., the Brazilian one. As we vary the capital ratio, the point where the resilience measure becomes negative, or otherwise said where the theoretical amplification explodes, closely predicts the point where a large cascade would ensue.
This illustrates how one, given an interbank network, could use the resilience measure and the theoretical amplification as a supervisory tool. This point is further developed in our follow-up paper on stress testing  \cite{amini10b}.

The crucial question in the context of macro prudential regulation of banking systems is how to identify and mitigate those features that make nodes systemically important.
We have  identified institutions acting as potential hubs for default contagion as those highly connected and with a large fraction of  {\it contagious links}.
One natural way to mitigate the systemic impact of these nodes is to set minimal capital requirements with respect to contagious links. This point of view is different from the current capital requirements as defined by the Basel II accords. Currently, minimal capital depends on the risk weighted sum of exposures, where the risk weights are given by the counterparty default probability. However, the default probability is computed by internal models and may not take into account knock-on effects. Our results suggest that, for financial stability, minimal ratios of capital should be set with respect to contagious exposures.

While the insolvency contagion investigated in this paper has been mostly associated with balance-sheet contagion, our results may be applied to players in the  over-the-counter derivatives markets. Contagion is carried in these markets  through intermediaries with a large fraction of critical receivables defined similarly to contagious links, i.e., receivables on which the intermediary depends to meet its own payment obligations  \cite{mincacont10}.  We argue that financial stability would be significantly enhanced by setting lower bounds on liquid reserves with respect to critical receivables, for those nodes which are counterparties to a large number of contracts.

Last, whereas in this paper we were concerned mostly with a network where the node's characteristics were observable, our results apply to the particular case where connectivities and weights are sequences of (exchangeable) random variables with arbitrary correlation structure. This corresponds to a setting where the modeler cannot observe the sequence of exposures, but rather has some belief over their distribution.
Such a perspective would have for example a market participant who models the default probability of its counterparties and takes into consideration network effets.



\paragraph{Acknowledgements.}
This work was presented at the MITACS Workshop on Financial Networks and Risk Assessment (Toronto, May 2010), the 6th Bachelier World Congress (Toronto, June 2010), the Workshop on Systemic Risk and Central Counterparties (Paris, Sept 2010), the Conference on Modeling and Managing Financial Risks (Paris, January 2011), 7th International Congress on Industrial and Applied Mathematics (Vancouver, July 2011) and Workshop on Econophysics of Systemic Risks and Network Dynamics (Kolkata, October 2011). We thank the participants of these conferences for helpful comments. We also thank D. Bienstock,  M. Crouhy, J. Gleeson for helpful discussions. Andreea Minca's work was supported by a doctoral grant from the Natixis Foundation for Quantitative Research.

\bibliographystyle{abbrv}
\bibliography{Biblio}  

\appendix
\section{Appendix: Proofs}\label{sec:proofs}
In this section we present the proofs of Theorem \ref{thm-main} and \ref{thm-resilient}. We begin by introducing a weighted configuration model --a multigraph related to the financial network-- which has the same asymptotic behavior as the random financial network.
We then show by  a coupling argument that the default cluster in the weighted configuration model can be constructed sequentially.
The contagion process in this sequential model may then be described by a  Markov chain.
We then generalize the differential equation method of Wormald \cite{Worm95} to the case where the dimension of the Markov chain depends on size of the network and we show that, as the network size increases, the rescaled Markov chain converges in probability to a limit described by a system of ordinary differential equations.
We solve these equations and obtain an analytical result on the final fraction of defaults in the network.
Finally, we show that, if the resilience measure is negative, the skeleton of contagious links percolates.

Consider a sequence $(\mathbf{e}_n,\gamman)_{n\geq 1}$ of financial networks satisfying Assumptions \ref{cond} and \ref{condition_threshold}.

\subsection{Link with the configuration model}
A standard method for studying  random graphs with prescribed
degree sequence is to consider (see e.g., \cite{bollobas, Molloy98thesize, janson08}) a
related random multigraph with the same degree sequence, known as the  {\em configuration model} \cite{bollobas},  then condition on this multigraph being simple.
The configuration model in the case of random directed graphs has been studied by Cooper and Frieze~\cite{coopfri04}.
Proceeding analogously, we introduce a  multigraph with the same degrees and exposures as the network defined above, but which is easier to study because of the independence properties of the variables involved.
Conditioned on being  a simple graph, it has the same law as the random financial network defined above.
\begin{definition}[Configuration model]\rm
\label{WeightedConfigModel}
Given a set of nodes $[n]=\{1, \dots, n\}$ and a degree sequence $(\bfdplusn, \bfdminusn)$,
we associate to each node $i$ two sets: $H_n^+(i)$ representing its out-going half-edges and $H_n^-(i)$ representing its in-coming half-edges, with $|H_n^+(i)|=\dplusn(i)$ and $|H_n^-(i)|=\dminusn(i)$.
Let $H_n^+ = \bigcup_i H_n^+(i)$ and $H_n^- = \bigcup_i H_n^-(i)$.
A \emph{configuration} is a matching of $H_n^+$ with $H_n^-$.
To each configuration we assign a graph. When an out-going half-edge of node $i$ is matched with an in-coming half-edge of node $j$, a directed edge from $i$ to $j$ appears in the graph.
The \emph{configuration model} is the random directed multigraph $\CM$ which is uniformly distributed across all configurations (Figure \ref{CM}).
\end{definition}

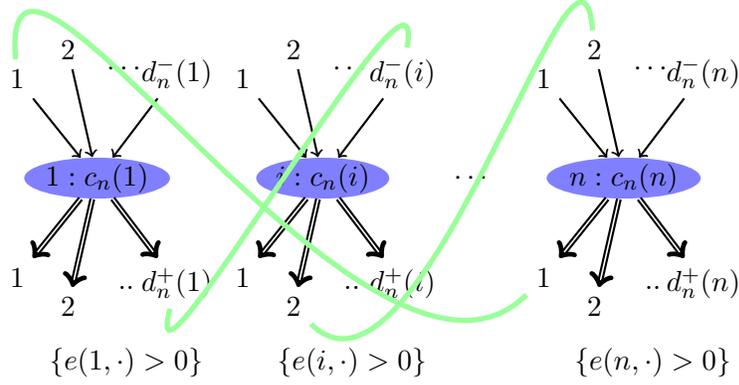
\begin{figure}
\begin{center}
\begin{tikzpicture}[scale = 1]
  \foreach \name/\x in {1/0, i/3, n/7}{
  \node[large] (G-\name) at (\x,5) {$\name : c_n(\name)$}
    child [grow=45]
    { node (G-\name3) [above]{$d_n^-(\name)$} edge from parent[inverse edge]}
    child [grow=75]
    { node {$\cdots$} edge from parent[draw = none]}
    child [grow=105]
    { node (G-\name2) [above]{$2$} edge from parent[inverse edge]}
    child [grow=135]
    { node (G-\name1) [above]{$1$} edge from parent[inverse edge]}
    child [grow=-45]
    { node (G-\name6) [below]{$d_n^+(\name)$} edge from parent[edge]}
    child [grow=-75]
    { node  (G-dots\name) {$..$} edge from parent[draw = none]}
    child [grow=-105]
    { node  (G-\name5) [below] {$2$} edge from parent[edge]}
    child [grow=-135]
    { node (G-\name4) [below]{$1$} edge from parent[edge]};
  }
  \node  at (5, 5) {$\cdots$};

  \foreach \name in {1,i, n}
    \node (G-weights\name) [below of = G-dots\name] {$\{e(\name, \cdot) > 0\}$};

    \draw[random edge] (G-11) to[out = 95, in = -135]  (G-n4);
    \draw[random edge] (G-i3) to[out = 80, in = -105]  (G-16);
    \draw[random edge] (G-n2) to[out = 95, in = -45]  (G-i5);

\end{tikzpicture}
\end{center}
\caption{Configuration model}
\label{CM}
\end{figure}
It is   easy to see that, conditional on  being a simple graph,
$\CM$  is    uniformly distributed on
$\mathcal{G}_n(\mathbf{e}_n)$. Thus,
 the law of $\CM$ conditional on  being a simple graph is the same as the law of $\mathbf{E}_n$.

In particular any property that holds with high probability (with probability tending to $1$ as $n \to \infty$) for the random multigraph $\CM$, holds with high probability on the random network
 $\mathbf{E}_n$ provided
\begin{equation}
 \liminf_{n \to \infty}\PP(\CM \ \mbox{is simple}) > 0.
 \label{eq:condSimple}
 \end{equation}

In particular, the condition $\sum_{i=1}^n (\dplusn(i))^2 + (\dminusn(i))^2 = O(n)$
implies (\ref{eq:condSimple}), see \cite{janson06}.

\begin{remark}\rm
Janson \cite{janson06} has studied, in the case of undirected graphs, the probability of the random multigraph  to be simple. One can adapt the  proof to the directed case and show that the condition $\sum_{i=1}^n (\dplusn(i))^2 + (\dminusn(i))^2 = O(n)$ implies (\ref{eq:condSimple}).  Indeed, in the non-directed case, Janson \cite{janson06} proves that when $m_n:= \sum_{i=1}^n d_n(i) \rightarrow \infty$, ($d_n(i)$ is the degree of node $i$) one has
\begin{eqnarray*}
\PP(G^*(n, (d_n(i))_1^n) \mbox{ is simple}) = \exp \left(-\frac{1}{2} \sum_i \lambda_{ii} -\sum_{i<j}(\lambda_{ij} - \log(1+\lambda_{ij})) \right) + o(1),
\end{eqnarray*}
where for $1 \leq i,j \leq n$; $\lambda_{ij}:= \frac{\sqrt{d_n(i)(d_n(i)-1)d_n(j)(d_n(j)-1)}}{m_n}$. The proof of these results is based on counting vertices with at least one loop and, pairs
of vertices with at least two edges between them, disregarding the number of parallel loops or edges. The same argument applies to the directed case, and one can show that when $m_n:= \sum_{i=1}^n d_n^{+}(i) =  \sum_{i=1}^n d_n^{-}(i)\rightarrow \infty$, then
\begin{eqnarray*}
\PP(\CM \mbox{ is simple}) = \exp \left(-\frac{1}{2} \sum_i \lambda_{ii} -\sum_{i<j}(\lambda_{ij} - \log(1+\lambda_{ij})) \right) + o(1),
\end{eqnarray*}
where for $1 \leq i,j \leq n$; $\lambda_{ij} = \frac{\sqrt{\dplusn(i)\dminusn(i)\dplusn(j)\dminusn(j)}}{m_n}$.
\end{remark}



One can observe that a uniform matching of half-edges can be obtained
sequentially: choose an in-coming half-edge according to any rule
(random or deterministic), and then choose the corresponding
out-going half-edge uniformly over the unmatched out-going
half-edges. The configuration model is thus particularly
appropriate for the study of contagion, as we will see in the
proofs, since we can restrict the matching process to choosing only in-coming half-edges entering defaulted nodes.
In doing so, one constructs directly the contagion cluster in the random graph given by the configuration model and endowed with the sequence of capital ratios.

 Due to this property, it is easier to study contagion on $\CM$ under
conditions on the degree sequence for the assumption above
(\ref{eq:condSimple}) to hold, then translate all results holding
with high probability to the initial network $\mathbf{E}_n$ defined in Definition
\ref{def:random}.

\subsection{Coupling } \label{sec-coupling}
We are given the set of nodes $[n]$ and their sequence of degrees $(\bfdplusn, \bfdminusn)$.  For each node $i$, we fix an indexing of its out-going and in-coming half-edges, ranging in $[\dplusn(i)] = \{1,\dots, \dplusn(i)\}$ and $[\dminusn(i)]$ respectively. Furthermore, all out-going half-edges are given a global label in the range $[1, \dots, m_n]$, with $m_n$ the total number of out-going (in-coming) half-edges. Similarly, all in-coming half-edges are given a global label in the range $[m_n]$.

Recall that $\Sigma_n(i)$ denotes the set of all permutations of the counterparties of $i$ in the network $\mathbf{e}_n$. The set of weights on the links exiting node $i$ is given by 
\begin{equation}
W_n(i) := \{e_n(i, j) > 0\}.
\label{eq:setWeights}
\end{equation}
For the sequence of edge weights and capital ratios, $(\mathbf{e}_n, \gamman)$, we generate the random graph $\CMM$, by the following algorithm:
\begin{enumerate}
\label{algorithm1}
\item For each node $i$, choose a permutation $\tau_n^i \in \Sigma_n(i)$ uniformly at random among all permutations of node $i$'s out-going half edges.
\item Color all in-coming and out-going half-edges in black. Define the set of initially defaulted nodes $$\mathcal{D}_0 := \bigcup_{i, \gamman(i) = 0}\{i\}.$$
    Set for all nodes $i \in [n] \backslash \mathcal{D}_0$;  $c(i) = \gamman(i)\sum_{j \in [n]} e_n(i,j)$.
\item
    At step $k \geq 1$, if the set of in-coming black half-edges belonging to nodes in $\mathcal{D}_{k-1}$ is empty, denote $\mathcal{D}_f$ the set $\mathcal{D}_{k-1}$. Otherwise:
    \begin{enumerate}
        \item Choose among all in-coming black half-edges of the nodes in $\mathcal{D}_{k-1}$ the in-coming half-edge with the lowest global label and color it in red.
            \label{stepa}
        \item
        \label{stepb}
        Choose a node $i$ with probability proportional to its number of black out-going half-edges and set $\pi_n(k) = i$. Let $i$ have $l-1$ out-going half-edges colored in red. Choose its $\tau_n^i(l)$-th out-going half-edge and color it in red. Let its weight be $w$. If the node $i\notin\mathcal{D}_{k-1}$ and $(1 - R)w$ is larger than $i$'s remaining capital then $\mathcal{D}_k = \mathcal{D}_{k-1} \bigcup \{i\}$. Otherwise, the capital of node $i$ becomes $c(i) - (1-R)w$.
        \item Match node $i$'s $\tau_n^i(l)$-th out-going half-edge to the in-coming half-edge selected at step (\ref{stepa}) to form an edge.
    \end{enumerate}
    \label{stepk}
\item Choose a random uniform matching of the remaining out-going half-edges and match them to the remaining in-coming half-edges in increasing order and color them all in red.\label{stepR}
\end{enumerate}

\begin{lemma}
The random graph $\CMM$ has the same distribution as
$\CM$. Furthermore the set $\mathcal{D}_f$ at the end of the above algorithm is the final set of defaulted nodes in the graph $\CMM$ (endowed with capital ratios $\mathbf{\gamma}_n$).
\label{lemmaequiv1}
\end{lemma}
\begin{proof}
The second claim is trivial. Let us prove the first claim. For a set $A$, we denote by $\Sigma_A$ the set of permutations of $A$. Let $\sigma^+_n$ and $\sigma^-_n$ be the random permutations in $\Sigma_{[m_n]}$, representing the order in which the above algorithm  selects the in-coming / out-going edges. At step $k$ of the above construction, in-coming half-edge with global label $\sigma^-_n(k)$ is matched to out-going half-edge with global label $\sigma^+_n(k)$ to form an edge.
The permutation $\sigma^+_n$ is determined by the set of permutations $(\tau_n^i)_{i = 1, \dots, n}$ and the sequence $\mathbf{\pi}_n$ of size $m_n$, representing the (ordered) sequence of nodes selected at Step \ref{stepb} (or Step \ref{stepR} when the set of in-coming black half-edges belonging to nodes in $\mathcal{D}_{k-1}$ is empty - assume we choose sequentially uniformly at random) of the algorithm (each node $i$ appears in sequence $\mathbf{\pi}_n$ exactly $\dplusn(i)$ times).

It is easy to see that $\sigma^+_n$ is a uniform permutation among all permutations in $\Sigma_{[m_n]}$, since $(\tau_n^i)_{i = 1, \dots, n}$ are uniformly distributed and at each step of the algorithm we choose a node with probability proportional to its black out-going half-edges.
On the other hand, the value of  $\sigma^-_n(k)$ depends in a deterministic manner on
$$(\mathbf{e}_n, \mathbf{\gamma}_n, \sigma^+_n(1), \dots, \sigma^+_n(k-1)) .$$ 



The out-going half-edge with global label $j$ is matched with the in-coming half-edge
with global label $(\sigma^-_n\circ(\sigma^+_n)^{-1})(j)$.
In order to prove our claim it is enough to prove that the permutation $(\sigma^-_n\circ(\sigma^+_n)^{-1})$ is uniformly distributed among all permutations of $m_n$.
Indeed, for an arbitrary permutation $\xi$ belonging to the set $\Sigma_{[m_n]}$, we have that
\begin{align*}
\mathbb{P} \left( \sigma^+_n(j) = \xi^{-1}(\sigma^-_n(j)) | \ \sigma^+_n(k) = \xi^{-1}(\sigma^-_n(k)) \ \mbox{for all} \ k < j \right) = \frac{1}{m_n - j + 1}.
\end{align*}
Indeed, conditional on the knowledge of $(\sigma^+_n(1), \dots, \sigma^+_n(j-1))$, $\sigma^-_n(j)$ is deterministic. Also, by conditioning on  $\forall k<j, \ \sigma^+_n(k) = \xi^{-1}(\sigma^-_n(k))$, then $\xi^{-1}(\sigma^-_n(j))\in \mathcal{T} := [m_n] \backslash \{\sigma^+_n(1), \dots, \sigma^+_n(j-1)\},$ of cardinal $m_n - j + 1$. In the above algorithm, $\sigma^+_n(j)$ has uniform law over $\mathcal{T}$. Then the probability to choose $\xi^{-1}(\sigma^-_n(j))$ is $\frac{1}{m_n - j + 1}$.

By the law of iterated expectations, we obtain that
$$\mathbb{P}(\sigma^-_n\circ(\sigma^+_n)^{-1} = \xi) = \mathbb{P}(\sigma^+_n = \xi^{-1}\circ \sigma^-_n) = \frac{1}{m_n!}.$$
This and the fact that the last step of the algorithm is a conditionally uniform match conclude the proof.
\end{proof}

\begin{center}\it
We can find the final set of defaulted nodes $\mathcal{D}_f$  of the above algorithm in the following manner: once the permutation $\tau_n^i$ is chosen, assign to each node its corresponding threshold $\theta_n(i) = \Theta_n(i, \tau_n^i)$ as in Definition \ref{thresholdfunction}, and forget everything about  $(\mathbf{e}_n, \mathbf{\gamma}_n)$.
\end{center}
\begin{definition}\rm
Denote by $\tilde{G}_n(\bfdplusn, \bfdminusn, \mathbf{\theta}_n)$ the random graph resulting from the above algorithm, in which we replace Step \ref{stepb} of the algorithm by the fact that node $i$ defaults the first time it has $\theta_n(i)$ out-going half-edges colored in red, i.e., at step
$$\inf\{k \geq 1, \mbox{ such that } \theta_n(i) = \#\{1\leq l \leq k, \ \pi_n(l) = i\}\}.$$
\end{definition}

\begin{corollary}
 The random graph $\tilde{G}_n(\bfdplusn, \bfdminusn, \mathbf{\theta}_n)$ has the same law as the unweighted skeleton of $\CMM$.
\end{corollary}
Let $N_n(j,k,\theta)$ denote the number of nodes with degree $(j,k)$ and threshold $\theta$ after choosing uniformly the random permutations $\mathbf{\tau}_n$ in the above construction.
\begin{lemma}\label{lem-init-frac}
We have (as $n\to\infty$)
$$\frac{N_n(j,k,\theta)}{n}  \stackrel{p}{\rightarrow}  \mu(j,k) p(j,k,\theta).$$
\end{lemma}
\begin{proof}
For any node $i$ with with degree $(j,k)$, the probability that its default threshold
$\Theta_n(i, \tau_n^i)$ be equal to $\theta$ is
$$\nu_n(i, \theta) := \frac{\# \{\tau \in \Sigma_n(i) \mid \ \Theta_n(i, \tau) = \theta  \}}{j!} .$$
Then we have
$$N_n(j,k,\theta) = \sum_{i, \ \dplusn(i) = j, \ \dminusn(i) = k} \Be(\nu_n(i, \theta)), $$
where $\Be(\cdot)$ denotes a Bernoulli variable.
\noindent By Assumption \ref{condition_threshold} we have
\begin{align*}
\EE[N_n(j,k,\theta) / n ] &= \mu_n(j,k) p_n(j,k,\theta) \stackrel{n \to \infty}{\rightarrow} \mu(j,k) p(j,k,\theta),
\end{align*}
\begin{align*}
{\rm and}\qquad{\rm Var}[N_n(j,k,\theta) / n] &= \frac{\sum_{i, \ \dplusn(i) = j, \ \dminusn(i) = k} \nu_n(i, \theta)(1 - \nu_n(i, \theta))} {n^2} \stackrel{n \to \infty}{\rightarrow} 0 .
\end{align*}
Now it is easy to conclude the proof by Chebysev's inequality.

\end{proof}

\subsection{A Markov chain description of contagion dynamics}\label{sec-markov}

In the previous section, we have replaced the description based on default rounds by an equivalent one based on successive bilateral interactions. By \emph{interaction} we mean matching an in-coming edge with an out-going edge. At each step of the algorithm described in last section, we have one interaction only between two nodes (banks), yielding at most one default. This allows for a simpler Markov chain which leads to the same set of final defaults.

We describe now the contagion process on the unweighted graph $\tilde{G}_n(\bfdplusn, \bfdminusn, \mathbf{\theta}_n)$ with thresholds $(\theta_n(i) = \Theta_n(i, \tau_n^i))_{1 \leq i \leq n}$ in terms of the dynamics of a Markov chain.

At each iteration we partition the nodes according to their state of solvency, degree, threshold and number of defaulted neighbors. Let us define
$S_{n}^{j, k, \theta, l}(t)$, the number of solvent banks with degree $(j,k)$, default threshold $\theta$ and $l$ defaulted debtors before time $t$.
We introduce the additional variables of interest:
\begin{itemize}
\item $D_n^{j, k, \theta}(t)$: the number of defaulted banks at time $t$ with degree $(j,k)$ and default threshold $\theta$,
\item $D_n(t)$: the number of defaulted banks at time $t$,
\item $D_n^-(t)$: the number of black in-coming edges belonging to defaulted banks,
\end{itemize}
for which it is easy to see that the following identities hold:
\begin{eqnarray*}
D_n^{j, k, \theta}(t) &=& \mu_n(j,k)p_n(j,k,\theta) - \sum_{0 \leq l < \theta} S_{n}^{j, k, \theta, l}(t), \\
D_n^-(t) &=& \sum_{j,k,0 \leq \theta \leq j} k D_n^{j, k,\theta}(t) - t ,\\
D_n(t) &=& \sum_{j,k,0 \leq \theta \leq j} D_n^{j, k,\theta}(t) .
\end{eqnarray*}
Because at each step we color in red one out-going edge and the number of black out-going edges at time $0$ is $m_n$, the number of black out-going edges at time $t$ will be $m_n-t$ .

By construction, $\mathbf{Y}_n(t) = \left( S_{n}^{j, k, \theta, l}(t) \right)_{j,k,0 \leq l < \theta \leq j}$ represents a Markov chain.
Let $(\mathcal{F}_{n, t})_{t \geq 0}$ be its natural filtration.
We define the operator $\wedge$ as
$$x \wedge y = \max(x,y) .$$
The length of the default cascade  is given by
\begin{equation}
T_n = \inf\{0 \leq t \leq m_n , \ D_n^-(t)=0\} \wedge m_n, \label{eq:Tf}
\end{equation}
The total number of defaults is given by $D_n(T_n)$, which represents the cardinal of the final set of defaulted nodes.

Let us now describe the transition probabilities of the Markov chain.
For $t < T_n$, there are three possibilities for the partner $B$ of an in-coming edge of a defaulted node $A$ at time $t+1$:
\begin{enumerate}
\item $B$ is in default, the next state is $\mathbf{Y}_n(t+1)=\mathbf{Y}_n(t) .$
\item $B$ is solvent, has  degree $(j,k)$ and default threshold $\theta$ and this is the $(l+1)$-th deleted out-going edge and $l+1 < \theta$. The probability of this event is $\frac{ (j - l) S_{n}^{j, k, \theta, l}(t)}{m_n-t}$. The changes for the next state will be
\begin{align*}
S_{n}^{j, k, \theta, l}(t+1) &= S_{n}^{j, k, \theta, l}(t) - 1 ,\\
S_n^{j, k, \theta, l+1}(t+1) &= S_n^{j, k, \theta, l+1}(t) + 1 .
\end{align*}

\item $B$ is solvent, has degree $(j,k)$ and default threshold $\theta$ and this is the $\theta$-th deleted out-going edge. Then with probability $\frac{ (j - \theta + 1) S_n^{j, k, \theta, \theta-1}(t)}{m_n-t}$ we have
\begin{align*}
S_n^{j, k, \theta, \theta-1}(t+1) = S_n^{j, k, \theta, \theta-1}(t) -1.
\end{align*}
\label{markovtransition}
\end{enumerate}
Let $\Delta_t$ be the difference operator: $\Delta_t Y := Y(t+1)-Y(t)$. We obtain the following equations for the expectation of $\mathbf{Y}_n(t+1)$, conditional on $\mathcal{F}_{n,t}$, by averaging over the possible transitions:
\begin{eqnarray}
\EE\left[\Delta_t S_n^{j, k,\theta, 0} |\mathcal{F}_{n,t} \right] &=& - \frac{j S_n^{j, k,\theta, 0}(t)}{m_n-t} , \nonumber \\
\EE\left[\Delta_t S_n^{j, k,\theta, l}|\mathcal{F}_{n,t}\right] &=&
\frac{(j-l+1) S_n^{j, k,\theta, l-1}(t)}{m_n-t} - \frac{(j-l) S_n^{j, k,\theta, l}(t)}{m_n-t}.
\label{eq:transitions}
\end{eqnarray}
The initial condition is
\begin{eqnarray*}
S_n^{j, k,\theta, l}(0) = N_n(j,k, \theta) \ind(l=0) \ind(0<\theta \leq j).
\end{eqnarray*}

\begin{remark}\rm
\label{remNegative}
We are interested in the value of $D_n(T_n)$, with $T_n$  defined in  \eqref{eq:Tf}.  In case $T_n < m_n$, the Markov chain can still be well defined for $t \in [T_n, m_n)$ by the same transition probabilities. However, after $T_n$ it will no longer be related to the contagion process and the value $D_n^-(t)$, representing for $t \leq T_n$ the number of in-coming half-edges belonging to defaulted banks, becomes negative.
We consider from now on that the above transition probabilities hold for $ t < m_n$.
\end{remark}

We will show in the next section that the trajectory of these variables for $t \leq T_n$
is close to the solution of the deterministic differential equations suggested by equations (\ref{eq:transitions}) with high probability.

\subsection{A law of large numbers for the contagion process}
Define the following set of differential equations denoted by ({\rm DE}):
\begin{eqnarray*}
(s^{j, k,\theta, 0})'(\tau) &=& - \frac{j s^{j, k,\theta, 0}(\tau)}{\lambda-\tau},\\
(s^{j, k,\theta, l})'(\tau) &=& \frac{(j-l+1) s^{j, k,\theta, l-1}(\tau)}{\lambda-\tau} - \frac{(j-l) s^{j, k,\theta, l}(\tau)}{\lambda-\tau},\qquad ({\rm DE}),
\end{eqnarray*}
with initial conditions
\begin{eqnarray*}
s^{j, k,\theta, l}(0)&=&\mu(j,k)p(j,k, \theta) \ind (l=0) \ind (0<\theta\leq j).
\end{eqnarray*}
\begin{lemma}\label{lem-sol}
The system of differential equations ({\rm DE}) admits the unique solution
$$ y(\tau) := \left(s^{j, k,\theta, l}(\tau)\right)_{j, k, 0 \leq l < \theta \leq j},$$
in the interval $0 \leq \tau < \lambda$, with
\begin{equation}
 s^{j, k,\theta, l}(\tau) :=  \mu(j,k) p(j, k, \theta){j \choose l} (1 - \frac{\tau}{\lambda})^{j-l}(\frac{\tau}{\lambda})^l\ind_{\{0<\theta\leq j\}}.
 \label{solutions}
\end{equation}
\end{lemma}
\begin{proof}

We denote by ${\rm DE}^K$ the set of differential equations defined above, restricted to $j \wedge k < K$ and by $b(K)$ the dimension of the restricted system.
Since the derivatives of the functions $\left(s^{j, k,\theta, l}(\tau)\right)_{j \wedge k < K, 0 \leq l < \theta \leq j}$ depend only on $\tau$ and the same functions, by a standard result in the theory of ordinary differential equations   \cite[Ch.2, Thm 11]{Hurewicz58}, there is an unique solution of ${\rm DE}^K$ in any domain of the type $(-\epsilon, \lambda) \times R$, with $R$ a bounded subdomain of $\RR^{b(K)}$ and $\epsilon > 0$.
The solution of $({\rm DE})$ is defined to be the set of functions solving all the finite systems $({\rm DE}^K)_{K \geq 1}$.

We solve now the system ${\rm DE}$.
Let $u = u(\tau)=- ln (\lambda - \tau)$. Then $u(0) = - ln (\lambda) $, $u$ is strictly monotone and so is the inverse function $\tau=\tau(u)$. We write the system of differential equations ({\rm DE}) with respect to $u$:
\begin{eqnarray*}
(s^{j, k,\theta, 0})'(u) &=& - j s^{j, k,\theta, 0}(u),\\
(s^{j, k,\theta, l})'(u) &=& (j-l+1) s^{j, k,\theta, l-1}(u) - (j-l) s^{j, k,\theta, l}(u).
\end{eqnarray*}
Then we have
\begin{eqnarray*}
\frac{d}{du} (s^{j, k, \theta, l+1}e^{(j-l-1)(u-u(0))}) = (j-l) s^{j, k,\theta, l}(u)
e^{(j-l-1)(u-u(0))} ,
\end{eqnarray*}
and by induction, we find
\begin{eqnarray*}
s^{j, k,\theta, l}(u) = e^{-(j-l)(u-u(0))}  \sum_{r=0}^l {{j-r}\choose{l-r}} \left( 1 - e^{-(u-u(0))} \right)^{l-r} s^{j, k,\theta, r}(u(0)).
\end{eqnarray*}
By going back to $\tau$, we have
\begin{eqnarray*}
s^{j, k,\theta, l}(\tau) = (1 - \frac{\tau}{\lambda})^{j-l} \sum_{r=0}^l s^{j, k,\theta, r}(0) {{j-r}\choose{l-r}}(\frac{\tau}{\lambda})^{l-r}.
\end{eqnarray*}
Then, by using the initial conditions, we find
$$s^{j, k,\theta, l}(\tau) = \mu(j,k) p(j, k,\theta) {j \choose l} (1 - \frac{\tau}{\lambda})^{j-l} (\frac{\tau}{\lambda})^l \ind_{\{\theta > 0\}}.$$
\end{proof}
 A key idea is to approximate, following Wormald   \cite{Worm95}, the
  Markov chain by the solution of a system of differential equations in the large network limit ~\cite{Worm95,Molloy98thesize}. We summarize here the main result of \cite{Worm95}.

For a set of variables $Y^1, ..., Y^b$ and for $U \subset \RR^{b+1}$,
define the stopping time $T_U=T_U(Y^1, ..., Y^b)=\inf\{t\geq 1, (t/n; Y^1(t)/n, ..., Y^b(t)/n) \notin U\}$.
\begin{lemma}[ Theorem 5.1. in \cite{Worm95}]
\label{thm-eqdif1}
Let $b\geq 2$ be an integer and consider a sequence of real valued random variables $(\{Y_n^l(t)\}_{1 \leq l \leq b})_{t\geq 0}$ and its natural filtration $\mathcal{F}_{n,t}$. Assume that there is a constant $C_0 > 0$ such that $|Y_n^l(t)| \leq C_0 n$ for all $n$, $t \geq 0$ and $1 \leq l \leq b$.
For all $l \geq 1$ let $f_l : \RR^{b+1} \to \RR$ be functions and assume that for some bounded connected open set $U \subseteq \RR^{b+1}$ containing the
closure of
$$\{(0,z_1,...,z_{b}):\exists \ n \mbox{ such that }\PP(\forall \ 1\leq l\leq b, \ Y_n^l(0)=z_l n) \neq 0 \},$$
the following three conditions are verified:
\begin{enumerate}
  \item {\rm(Boundedness).} For some function $\beta(n) \geq 1$ we have for all $t<T_{U}$
  $$\max_{1 \leq l \leq b} |Y_n^l(t+1)-Y_n^l(t)|\leq \beta(n).$$
  \item {\rm(Trend).} There exists $\lambda_1(n) = o(1)$ such that for  $1 \leq l \leq b$ and $t<T_{U}$
  $$|\EE[Y_n^l(t+1)-Y_n^l(t)|\mathcal{F}_{n,t}]-f_l(t/n,Y_n^1(t)/n,...,Y_n^l(t)/n)| \leq \lambda_1(n) .$$
  \item {\rm(Lipschitz).}  The functions $(f_l)_{1 \leq l \leq b}$ are Lipschitz-continuous  on $U$.
\end{enumerate}
Then the following conclusions hold:
\begin{description}
  \item[(a)] For $(0,\hat{z}_1,...,\hat{z}_{b}) \in U$, the system of differential equations
  $$\frac{dz_l}{ds}=f_l(s,z_1,...,z_l), \ \ l=1,...,b ,$$ has a unique solution in $U$, $z_l:\RR \rightarrow \RR$, which passes through $z_l(0)=\hat{z}_l$, for $l=1,\dots,b$, and which extends to points arbitrarily close to the boundary of $U$.
  \item[(b)] Let $\lambda>\lambda_1(n)$ with $\lambda=o(1)$.
  For a sufficiently large constant C, with probability $1-O\left(\frac{b\beta(n)}{\lambda} \exp \left( - \frac{n\lambda^3}{\beta(n)^3}\right)\right)$, we have
 $$\sup_{0 \leq t \leq \sigma(n) n}(Y_n^l(t) - n z_n^l(t/n)) = O(\lambda n),$$ where $\mathbf{z}_n(t)=(z_n^1(t),\dots,z_n^b(t))$ is the solution of  $$ \frac{d\mathbf{z}_n}{dt}=f(t,\mathbf{z}_n(t))\qquad {z}_n(0) = \mathbf{Y}_n(0)/n$$  $${\rm and}\qquad\sigma(n)=\sup \{ t\geq 0,\quad d_{\infty}(\mathbf{z}_n(t),\partial U)\geq C\lambda\}.$$
\end{description}
\end{lemma}
 We apply this lemma to the contagion model described in Section \ref{sec-markov}.
Let us define, for $0 \leq \tau \leq \lambda$
\begin{eqnarray*}
\delta^{j, k,\theta}(\tau) &:=&  \mu(j,k) p(j, k, \theta) - \sum_{0 \leq l < \theta} s^{j, k, \theta, l}(\tau),\\
\delta^-(\tau) &:=&  \sum_{j,k,\theta} k \delta^{j,k,\theta}(\tau) - \tau , \ \ \mbox {and}  \\
\delta(\tau) &:=& \sum_{j,k,\theta} \delta^{j,k,\theta}(\tau),
\end{eqnarray*}
with $s^{j, k, \theta, l}$ given in Lemma \ref{lem-sol}.
With $\Bin(j, \pi)$ denoting a binomial variable with parameters $j$ and $\pi$, we have
\begin{eqnarray}
\delta^{j, k,\theta}(\tau) &=&  \mu(j,k) p(j, k, \theta) \PP\left(\Bin(j,\frac{\tau}{\lambda}) \geq \theta \right) ,\\
\delta^-(\tau) &=& \sum_{j,k,\theta} k \delta^{j,k,\theta}(\tau) - \tau \nonumber\\
&=&  \sum_{j,k, \theta \leq j} k \mu(j,k) p(j, k, \theta)\PP\left(\Bin(j,\frac{\tau}{\lambda}) \geq \theta \right)  - \tau \label{eq:delt}\\
&=& \lambda(I(\frac{\tau}{\lambda}) - \frac{\tau}{\lambda})\nonumber,
\end{eqnarray}
and
\begin{eqnarray}
\delta(\tau) &:=& \sum_{j,k,0 \leq \theta \leq j} \mu(j,k) p(j, k, \theta) \PP\left(\Bin(j,\frac{\tau}{\lambda}) \geq \theta \right).
\label{eq:delta}
\end{eqnarray}

\subsection{Proof of Theorem \ref{thm-main}}\label{sec-proof-main}
We now proceed to the proof of Theorem \ref{thm-main} whose aim is to approximate the value $D_n(T_n)/n$ as $n\to \infty$.
We base the proof on Theorem \ref{thm-eqdif1}.  However, several difficulties arise since in our case since the number of variables depends on $n$.
We first need to bound the contribution of higher order terms in the infinite sums (\ref{eq:delt}) and (\ref{eq:delta}). Fix  $\epsilon>0$.
By Condition \ref{cond}, we know
\begin{eqnarray*}
\lambda = \sum_{j,k} k\mu(j,k) = \sum_{j,k} j\mu(j,k) \in (0,\infty).
\end{eqnarray*}
Then, there exists an integer $K_{\epsilon}$, such that $$\sum_{k \geq K_{\epsilon}}\sum_{j} k\mu(j,k) + \sum_{j\geq K_{\epsilon}} \sum_{k} j \mu(j,k) < \epsilon ,$$ which implies that
$$\sum_{j \wedge k \geq K_{\epsilon}} k\mu(j,k) <\epsilon .$$
It follows that
\begin{equation}
\forall \ 0 \leq \tau \leq \lambda, \sum_{j \wedge k \geq K_{\epsilon}, 0 \leq \theta \leq j} k \mu(j,k) p(j, k, \theta) \PP\left(\Bin(j,\frac{\tau}{\lambda}) \geq \theta \right) < \epsilon.
\label{limitdelta_}
\end{equation}
The number of vertices with degree $(j,k)$ is $n\mu_n(j, k)$. Again, by Condition \ref{cond}, $$\sum_{j,k} k \mu_n(j,k) = \sum_{j,k} j \mu_n(j,k) \rightarrow \lambda \in (0,\infty) .$$
Therefore, for $n$ large enough, $\sum_{j \wedge k \geq K_{\epsilon}} k\mu_n(j,k) < \epsilon ,$
and
\begin{equation}
\forall \ 0 \leq  t \leq m_n, \sum_{j \wedge k \geq K_{\epsilon}, 0 \leq \theta \leq j} kD_n^{j, k,\theta}(t)/n < \epsilon.
\label{limitDn}
\end{equation}
For $K \geq 1$, we denote
\begin{align*}
{\bf y}^K &:= \left(s^{j, k,\theta, l}(\tau)\right)_{j \wedge k < K, \ 0 \leq l < \theta \leq j} \mbox{ and } \\
Y_n^K &:= \left(S_n^{j, k,\theta, l}(\tau)\right)_{j \wedge k < K, \ 0 \leq l < \theta \leq j},
\end{align*}
 both of dimension   $b(K)$, where $\delta^{j, k,\theta}(\tau), s^{j, k,\theta, l}(\tau)$ are solutions to a system (${\rm DE}$) of ordinary differential equations.
Let $$\pi^* = \min\{\pi \in [0,1]|I(\pi) = \pi\}.$$
For the arbitrary constant $\epsilon > 0$ we fixed above, we define the domain $U_{\epsilon}$ as
\begin{align}
U_{\epsilon}=\{\left(\tau, y^{K_{\epsilon}} \right) \in \RR^{b(K_{\epsilon})+1} \ : \    -\epsilon <  \tau < \lambda - \epsilon \ , \ -\epsilon < s^{j, k,\theta, l} < 1\}.
\label{uepsilon}
\end{align}
The domain $U_{\epsilon}$ is a bounded open set which contains the support of all initial values of the variables. Each variable is bounded by a constant times $n$ ($C_0 = 1$).
By the definition of our process, the Boundedness condition is satisfied with $\beta(n) = 1$.
The second condition of the theorem is satisfied by some $\lambda_1(n)=O(1/n)$. Finally the Lipschitz property  is also satisfied since $\lambda-\tau$ is bounded away from zero. Then by Lemma~\ref{thm-eqdif1} and by using  Lemma ~\ref{lem-init-frac} for convergence of initial conditions, we have :

\begin{corollary} \label{lem-DE}
For a sufficiently large constant $C$, we have
\begin{eqnarray}
\label{eq:diffmethod}
\PP(\forall t \leq n \sigma_C(n), \mathbf{Y}_n^{K_{\epsilon}}(t)=n\mathbf{y}^{K_{\epsilon}}(t/n)+O(n^{3/4})) = 1 - O(b(K_{\epsilon})n^{-1/4}exp(-n^{-1/4}))
\end{eqnarray}
 uniformly for all $t \leq n\sigma_C(n)$ where
$$\sigma_C(n)=\sup\{\tau\geq 0,  d( \mathbf{y}^{K_{\epsilon}}(\tau),\partial U_{\epsilon}\ )\geq  Cn^{-1/4} \}.$$
\end{corollary}

\noindent When the solution reaches the boundary of $U_{\epsilon}$, it violates the first constraint in \ref{uepsilon}, determined by $\hat{\tau} = \lambda -\epsilon$.
By convergence of $\frac{m_n}{n}$ to $\lambda$, there is a value $n_0$ such that $\forall n \geq n_0$, $\frac{m_n}{n} > \lambda - \epsilon$, which ensures that $\hat{\tau}n\leq m_n$.
Using (\ref{limitdelta_}) and (\ref{limitDn}), we have, for $0 \leq t \leq n \hat{\tau}$ and $n \geq n_0$:
\begin{eqnarray}
\left| D_n^-(t)/n - \delta^-(t/n) \right| &=& | \sum_{j,k}\sum_{\theta \leq j}  k (D_n^{j, k,\theta}(t)/n-\delta^{j, k,\theta}(t/n)) | \nonumber\\
&\leq& \sum_{j,k}\sum_{\theta \leq j}  k \left| D_n^{j, k,\theta}(t)/n-\delta^{j, k,\theta}(t/n) \right| \nonumber\\
&\leq& \sum_{j \wedge k \leq K_{\epsilon}}\sum_{\theta \leq j}  k \left| D_n^{j, k,\theta}(t)/n-\delta^{j, k,\theta}(t/n) \right| + 2 \epsilon,
\label{limitDiff}
\end{eqnarray}
and
\begin{eqnarray}
\left| D_n(t)/n - \delta(t/n) \right| &\leq& \sum_{j \wedge k \leq K_{\epsilon}}\sum_{\theta \leq j}  \left| D_n^{j, k,\theta}(t)/n-\delta^{j, k,\theta}(t/n) \right| + 2 \epsilon,
\label{dn}
\end{eqnarray}
We obtain by Corollary \ref{lem-DE} that
\begin{eqnarray}
\sup_{t \leq \hat{\tau} n} \left| D_n^-(t)/n - \delta^-(t/n) \right| \leq 2\epsilon + o_p(1)\\
\sup_{t \leq \hat{\tau} n} \left| D_n(t)/n - \delta(t/n) \right| \leq 2\epsilon + o_p(1)
\label{convergence}
\end{eqnarray}
We nw study the stopping
time $T_n$  defined in  \eqref{eq:Tf} and the size of the default cascade $D_n(T_n)$.
 First assume  $I(\pi) > \pi$ for all $\pi \in [0,1)$, i.e., $\pi^*=1$. Then we have
$$\forall \tau < \hat{\tau},  \delta^-(\tau) = \sum_{j,k,\theta} k \delta^{j, k,\theta}(\tau)  - \tau > 0.$$
We have then that $T_n/n = \hat{\tau} + O(\epsilon) + o_p(1)$ and from  convergence (\ref{convergence}), since $\delta(\hat{\tau})= 1 - O(\epsilon)$, we obtain by tending $\epsilon$ to $0$ that $|D_n(T_n)| = n - o_p(n)$. This proves the first part of the theorem.

Now consider the case $\pi^* < 1$, and furthermore $\pi^*$ is a stable fixed point of $I(\pi)$. Then by definition of $\pi^*$ and by using the fact that $I(1) \leq 1$, we have $I(\pi)<\pi$
for some interval $(\pi^*,\pi^* + \tilde{\pi})$. Then $\delta^-(\tau)$ is negative in an interval $(\tau^*, \tau^* + \tilde{\tau}), $ with $\tau^* = \lambda\pi^*$.

Let  $\epsilon$   such that $2 \epsilon < -\inf_{\tau \in (\tau^*, \tau^* + \tilde{\tau})}\delta^-(\tau)$ and denote $\hat{\sigma}$ the first iteration at which it reaches the minimum.
Since $\delta^-(\hat{\sigma}) < -2 \epsilon$ it follows that with high probability $D^-(\hat{\sigma} n)/n < 0$,
so $T_n/n = \tau^*  + O(\epsilon) + o_p(1)$.
The conclusion follows by taking the limit $\epsilon\to 0$.

\subsection{Proof of Theorem \ref{thm-giant}} \label{proof:thm-geant}
Strong connectivity sparse random directed graphs with prescribed degree sequence has been studied by Cooper and Frieze in \cite{coopfri04}.
Let $\lambda_n$ represent the average degree (then by Condition~\ref{cond}, $\lambda_n \rightarrow \lambda$ as $n \to \infty$), and $\mu_{n}(j,k)$ represent the empirical distribution of the degrees, assumed to be proper (as defined below), then  \cite[Theorem 1.2]{coopfri04} states that if
\begin{equation}
\sum_{j, k} j k \frac{\mu(j,k)}{\lambda} > 1,
\label{condgiant}
\end{equation}
then the graph contains w.h.p. a strongly connected giant component.

We remark that the theorem above is given in \cite{coopfri04} under stronger assumptions on the degree sequence, adding to Assumption \ref{cond} the following three conditions, in which $\Delta_n$ denotes the maximum degree:
\begin{itemize}
\item Let $\rho_n = \max(\sum_{i, j}\frac{i^2 j \mu_n(i, j)}{\lambda_n}, \sum_{i, j}\frac{j^2 i\mu_n(i, j)}{\lambda_n})$. If $\Delta_n \to \infty$ with $n$ then $\rho_n = o(\Delta_n)$.
\item $\Delta_n \leq \frac{n^{1/12}}{\log n}$.
\item As $n \to \infty$, $\nu_n \to \nu \in (0, \infty)$.
\end{itemize}
Following \cite{coopfri04} we call a degree sequence proper if it satisfies Assumption \ref{cond} together with the above conditions.

A first reason for adding these conditions in \cite{coopfri04} is to ensure that Equation (\ref{eq:condSimple}) holds. However, following Janson \cite{janson06}, the restricted set of conditions \ref{cond} is sufficient.
The second reason is that \cite{coopfri04} gives more precise results on the structure of the giant component.
For our purpose, to find the sufficient condition for the existence of strongly connected giant component, we show that these supplementary conditions may be dropped.

 It is easy to see that a bounded degree sequence (i.e., $\Delta_n = O(1)$) which satisfies Assumption \ref{cond} is proper. We use this fact in the following.
\begin{lemma}\label{lem-cg}
Consider the random directed graph $\config$ constructed by configuration model, where the degree sequence satisfies Assumption \ref{cond}.
If
\begin{equation}
\sum_{j, k} j k \frac{\mu(j,k)}{\lambda} > 1,
\end{equation}
then with high probability the graph contains  a strongly connected giant component.
\end{lemma}
\begin{proof}
By the second moment property and Fatou's lemma,  there exists a constant $C$ such that
\begin{eqnarray*}
\sum_{j, k}jk \mu(j,k) &\leq& \sum_{j, k}(j^2+k^2)\mu(j,k) \\
&\leq& \liminf_{n\rightarrow \infty} \sum_{j, k}(j^2+k^2)\mu_n(j,k)\leq C.
\end{eqnarray*}
Then, it follows that for arbitrary $\epsilon > 0$, there exists a constant $\Delta_{\epsilon}$ such that
$$\sum_{j \wedge k > \Delta_{\epsilon}}jk \mu(j,k) \leq \epsilon.$$
Thus, by choosing $\epsilon$ small enough, there exists a constant $\Delta_{\epsilon}$ such that
$$\sum_{j \wedge k \leq \Delta_{\epsilon}}jk \frac{\mu(j,k)}{\lambda} > 1.$$
We now modify the graph such that the maximum degree is equal to $\Delta_{\epsilon}$: for every node $i$ such that $\dplus_n(i) \wedge \dminus_n(i) > \Delta_{\epsilon}$, all its in-coming (resp. out-going) half-edges are transferred to new nodes with degree $(0, 1)$ (resp. with degree $(1,0)$).
Since these newly created nodes cannot be part of any strongly connected component, it follows that, if the modified graph contains such a component, then necessarily the initial graph also does.
It is then enough to evaluate Equation (\ref{condgiant}) for this modified graph, which by construction verifies the
Assumption \ref{cond} for the new empirical distribution $\tilde{\mu}$ with the average degree $\tilde{\lambda}$. Also, since the degrees of the modified graph are bounded, the supplementary conditions above also hold, i.e., the degree sequence is proper, and we can apply Cooper \& Frieze's result.
It only remains to show that $\sum_{j, k} j k \frac{\tilde\mu(j,k)}{\tilde\lambda} > 1$.
Indeed, we have
\begin{eqnarray*}
\sum_{j, k} j k \frac{\tilde\mu(j,k)}{\tilde\lambda} &=& \sum_{j \wedge k \leq \Delta_{\epsilon}} j k \frac{\tilde\mu(j,k)}{\tilde\lambda} \\
&=& \sum_{0 < j,k \leq \Delta_{\epsilon}} j k \frac{\tilde\mu(j,k)}{\tilde\lambda}\\
&=& \sum_{0 < j,k \leq \Delta_{\epsilon}} j k \frac{\mu(j,k)}{\lambda} > 1.
\end{eqnarray*}
The last equality follows from the fact that for ${0 < j,k \leq \Delta_{\epsilon}}$, we have
$$\frac{\tilde\mu(j,k)}{\tilde\lambda}=\frac{\mu(j,k)}{\lambda}.$$
This is true since the total number of edges, and the number of nodes with degree $j,k$ for ${0 < j,k \leq \Delta_{\epsilon}}$, stays unmodified.

\end{proof}

We now proceed to the proof of Theorem \ref{thm-giant}.
Our proof is based on ideas applied in  \cite{Fountoulakis07, janson08} for site and bond percolation in configuration model.
Our aim is to show that the skeleton of contagious links in the random financial network is still described by configuration model, with a degree sequence verifying Assumptions \ref{cond}, and then apply Lemma \ref{lem-cg}.

For each node $i$, the set of
contagious out-going edges is given by $$C_n(i) := \{l \mid
(1 - R)e_n(i,l) > \gamman(i)\}.$$ Let us denote their number by $$c_n^{+}(i) := \#C_n(i).$$
We denote by $G^{c}_{n}$ the unweighted skeleton of contagious links in the random network $\CM$, endowed with the capital ratios $\mathbf{\gamma}_n$.

In order to characterize the law of  $G^{c}_{n}$, we adapt Janson's method \cite{janson08} for the directed case.
\begin{lemma}
\label{lem:expl}
The unweighted skeleton of contagious links $G^{c}_{n}$ has the same law as the random graph constructed as follows:
\begin{enumerate}
\label{algorithm2}
\item Replace the degree sequence $(\bfdplusn, \bfdminusn)$ of size $n$ by the degree sequence $(\mathbf{\tilde{d}}_{n'}^+, \mathbf{\tilde{d}}_{n'}^-)$ of size $n'$, with
    \begin{align*}
    n' &= n + m_n - \sum_{i = 1}^n c_n^{+}(i),\\
    \forall \ 1\leq i \leq n, &\ \tilde{d}_{n'}^+(i) = c_n^{+}(i), \ \tilde{d}_{n'}^-(i) = d_n^-(i),\\
    \forall \ n+1\leq i \leq n', &\ \tilde{d}_{n'}^+(i) = 1, \ \tilde{d}_{n'}^+(i) = 0.
    \end{align*}
\item
Construct the random unweighted graph $G^*_{n'}(\mathbf{\tilde{d}}_{n'}^+, \mathbf{\tilde{d}}_{n'}^-)$ with $n'$ nodes, and the degree sequence $(\mathbf{\tilde{d}}_{n'}^+, \mathbf{\tilde{d}}_{n'}^-)$ by configuration model.
\item Delete $n^+ = n' - n$ randomly chosen nodes
with out-degree $1$ and in-degree $0$.
\end{enumerate}
\end{lemma}

\begin{proof}
The skeleton $G^{c}_{n}$ can be obtained in a two-step  procedure. First, disconnect all non-contagious links in $\CM$ from their end nodes and  transfer them to newly created nodes of degree $(1, 0)$. Then delete all new nodes and their incident edges. The first step of this procedure may be dubbed as ``rewiring''.
Looking at graphs as configurations, and since the first step changes the total number of nodes but not the number of half-edges, it is easy to see that there is a one to one correspondence between the configurations before and after the ``rewiring''. Thus, the graph after rewiring is still described by the configuration model, and  has the same law as
$G^*_{n'}(\mathbf{\tilde{d}}_{n'}^+, \mathbf{\tilde{d}}_{n'}^-)$.
Finally, by symmetry, the nodes with out-degree $1$ and in-degree $0$ are equivalent, so one may remove randomly the appropriate number of them.
\end{proof}

 Note that since the degree sequence before rewiring verifies Condition \ref{cond}, so does the degree sequence after rewiring. Moreover, since we are interested in the strongly connected component and nodes of degrees $(1, 0)$ will not be included, we can actually apply Lemma \ref{lem-cg} to the random graph resulting by the above contagion process. Hence, we may study the strongly connected component in the intermediate graph $G^*_{n'}(\mathbf{\tilde{d}}_{n'}^+, \mathbf{\tilde{d}}_{n'}^-)$.

Let us denote by $l_{n'}(j,k)$, the number of nodes with out-degree $j$ and in-degree $k$ in the graph $G^*_{n'}(\mathbf{\tilde{d}}_{n'}^+, \mathbf{\tilde{d}}_{n'}^-)$, and
by $\tilde{\lambda}_{n'}$, the average degree. Then the average directed degree in this random graph is given by
$\nu_n := \sum_{j,k}jkl_{n'}(j,k)/(\tilde{\lambda}_{n'}n').$

We first observe that
$\tilde{\lambda}_{n'}n' = \lambda n$, since the number of edges is
unchanged after rewiring of the links. For every $k > 0$, the quantity
$\sum_jjl_{n'}(j,k)$ represents the number of out-going edges belonging to
nodes with in-degree $k$ in the graph after rewiring, which in
turn represents the number of contagious out-going edges belonging
to nodes with in-degree $k$ in the graph before rewiring. But so
does $\sum_j p_n(j,k,1)n\mu_n(j,k) j$.
So, for all $k$
\begin{eqnarray*}
\sum_jj\frac{l_{n'}(j,k)}{\lambda_{n'}n'} &=& \frac{1}{\lambda_{n'}n'}\sum_j p_n(j,k,1)n\mu_n(j,k) j \\ &=&
\sum_j j p_n(j,k,1)\frac{\mu_n(j,k)}{\lambda_n}\\
&\stackrel{n \to \infty}{\to}& \sum_j j p(j,k,1)\frac{\mu(j,k)}{\lambda},
\end{eqnarray*}
where convergence holds by the second moment property in Assumption \ref{cond}.
Applying Lemma \ref{lem-cg} to the sequence of degrees in the graph after rewiring shows that when
\begin{eqnarray*}
\sum_k k \lim_n \sum_jj\frac{l_{n'}(j,k)}{\lambda_{n'}n'} = \sum_k \sum_j j p(j,k,1)\frac{\mu(j,k)}{\lambda} > 1,
\end{eqnarray*}
then with high probability there exists a giant strongly connected component in the skeleton of contagious links.

\end{document}